\documentclass{llncs}

\usepackage[applemac]{inputenc}
\usepackage{authblk} %removed for compliance with Springer template

\usepackage{setspace}

\usepackage{booktabs,gensymb}
\usepackage[a4paper, margin=1in, bottom = 4cm, top = 4cm]{geometry} %removed for compliance with Springer template
\usepackage{bm,rotating}  %removed for compliance with Springer template
\usepackage[table]{xcolor} %removed for compliance with Springer template

\usepackage{amsmath, amssymb}
\usepackage[hidelinks]{hyperref} %removed for compliance with Springer template
\usepackage{mathtools} % for \coloneqq etc

\usepackage{tikz}
\usetikzlibrary{trees,matrix, arrows,fit, positioning, shapes.geometric}
\usepackage{colortbl} % for \rowcolor
\definecolor{lightgray}{gray}{0.9}

\newcolumntype{C}[1]{>{\centering\let\newline\\\arraybackslash\hspace{0pt}}m{#1}}

\usepackage{cleveref} %for \Cref
\usepackage{braket} % for \Set{}
\usepackage{thm-restate} % for restatable
\usepackage{orcidlink} % for orcid-id logo

\definecolor{mycolor1}{HTML}{1E3ECF}
\definecolor{mycolor2}{HTML}{41ED84}
\definecolor{mycolor3}{HTML}{D809A1}
\definecolor{mycolor4}{HTML}{F58B00}

\def\R{\mathcal{R}}
\def\r{\mathbf{r}}
\def\T{\mathcal{T}}
\def\agg{\mathbf{agg}}

\def\a{\mathtt{a}} % an action
\def\b{\mathtt{b}} % another action

\newcommand{\rlike}{\r^\mathrm{like}} % formerly r^1
\newcommand{\rrisk}{\r^\mathrm{risk}} % formerly r^3
\newcommand{\rneg}{\r^\mathrm{negl}} % formerly r^4
\newcommand{\Rlike}{\R^\mathrm{like}} % formerly R^1
\newcommand{\Rrisk}{\R^\mathrm{risk}} % formerly R^3
\newcommand{\Rneg}{\R^\mathrm{negl}} % formerly R^4

\newcommand{\x}{\mathbf{x}} % a vector

\newcommand*\ov[1]{\overline{#1}}

\newcommand{\kick}[1]{}

\title{An Axiomatic Approach to Formalized Responsibility Ascription}
\author{Sarah Hiller\inst{1,3}\footnote{Corresponding author. Contact: sarah.hiller@fu-berlin.de.}\orcidlink{0000-0002-4836-4036} \and
Jonas Israel\inst{2}\orcidlink{0000-0002-3992-3203} \and
Jobst Heitzig\inst{3}\orcidlink{0000-0002-0442-8077}}

\institute{Free University Berlin, Institute for Mathematics, Berlin, Germany \and
Technical University Berlin, Efficient Algorithms Research Group, Berlin, Germany \and
Potsdam Institute for Climate Impact Research, Potsdam, Germany}

\begin{document}

\maketitle

\begin{abstract}
    
    A formalized and quantifiable responsibility score is a crucial component in many aspects of the development and application of multi-agent systems and autonomous agents. 
    We can employ it to inform decision making processes based on ethical considerations, as a measure to ensure redundancy that helps us in avoiding system failure, as well as for verifying that autonomous systems remain trustworthy by testing for unwanted responsibility voids in advance. 
    We follow recent proposals to use probabilities as the basis for responsibility ascription in uncertain environments rather than the deterministic causal views employed in much of the previous formal philosophical literature. 
    Using an axiomatic approach we formally evaluate the qualities of (classes of) proposed responsibility functions. To this end, we decompose the computation of the responsibility a group carries for an outcome into the computation of values that we assign to its members for individual decisions leading to that outcome, paired with an appropriate aggregation function.
    Next, we discuss a number of intuitively desirable properties for each of these contributing functions. We find an incompatibility between axioms determining upper and lower bounds for the values assigned at the member level\kick{, imposing a choice for either the upper or the lower bound}. Regarding the aggregation from member-level values to group-level responsibility we are able to axiomatically characterise one promising aggregation function. 
    Finally, we present two maximally axiom compliant group-level responsibility measures -- one respecting the lower bound axioms at the member level and one respecting the corresponding upper bound axioms.
    
   
   % uncomment the following when published 
    %\blfootnote{\textbf{Preprint Note:} This preprint has not undergone peer review or any post-submission improvements or corrections. The Version of Record of this contribution is published in [insert volume title], and is available online at https://doi.org/[insert DOI] }
\end{abstract}

\section{Background and Introduction}\label{sec:intro}

% General motivation for investigating responsibility - reference to multi agent systems 
% Explanation of what we previously did and a quick overview of the literature (+ how we relate)
% Causation vs probability raising 
% Explanation of axiomatic method 

One important challenge in the development and deployment of trustworthy autonomous and multi-agent systems is a formalized and quantifiable responsibility measure.
Internally to a system we can use responsibility as a 
value which informs decision making processes, while from an external point of view we can use shared responsibility as a measure to ensure redundancy and avoid system failures in the case of an agent failure. Finally, in the evaluation of autonomous agents interacting with humans it is important to verify the requirement of meaningful human control \cite{santoni_de_sio_meaningful_2018}.
The study of responsibility in multi-agent systems lies at the intersection of game theory, formal ethics, logic, and computer science. For a recent overview of open questions regarding responsibility formalisation in autonomous systems applications see \cite{yazdanpanah_responsibility_2021}.
In the current paper we focus on backward looking responsibility for groups of agents in complex scenarios including uncertainty.

In order to arrive at a formalized responsibility representation we need to take two steps: First, we need to represent the situation in which responsibility is to be ascribed. The way this is implemented varies considerably between approaches from a philosophical tradition versus ones from a computer science background. Second, we need to look at the criteria for responsibility ascription and how these can be modelled in the given situation.

%\paragraph{Related Work.}
Braham and van Hees \cite{braham2012} model the decision situations in which responsibility is assigned using normal-form games. That is, several players may interact, but decisions are all made at one single point in time, after which the outcome is clear.
% stit
Contrastingly, in the field of \emph{stit} logics, decisions take place at a moment in an infinitely branching history. Every choice is represented by a new branching off of possible futures \cite{stit2001,horty2001}. 
Finally, Yazdanpanah et al. \cite{yazdanpanah_strategic_2019} use concurrent game structures where agents' actions initiate a change between finitely many possible states of a system, with the possibility of revisiting a state more than once. Beyond these, there are many other ways to model the relevant decision situations for the study of responsibility ascription (e.g., see \cite{BBB+18a,GlPa19a}) which are, however, less relevant to our approach.
% 
% our approach
Our representation is somewhere in between the previously mentioned approaches of using normal-form games, infinite Kripke frames, and concurrent games. Namely, we use extensions of extensive-form game trees. These extend normal-form games with a temporal component, or, viewed from a different angle they can be conceived of as finite fragments of infinite stit frames \cite{duijf2018}. They are perhaps most similar to the concurrent epistemic game structures used by \cite{yazdanpanah_strategic_2019}, with several interacting agents having the choice to influence the system. However, we include the option of consequences being probabilistic or uncertain, as agents may not always know about the consequences of their choices in interactive settings.

% CRITERIA F RESP: Causation
The criteria for responsibility ascription are mostly consistent throughout recent philosophical literature.
Free choice and the capability for moral reasoning \cite{vandepoel2011,vallentyne2008} are generally assumed to hold for all agents and actions (possibly after an appropriate restriction).
Most attention tends to be focused on the analysis of a causal relation between the action or event in question and an undesirable outcome.
Implementations of causality include NESS causation \cite{braham2012,wright1988}, or modelling causality as \emph{seeing to it that} a certain outcome obtains in stit logics \cite{stit2001}. Very similar to this strict view, \cite{yazdanpanah_strategic_2019} represent responsibility as the possibility to avoid a certain outcome, that is, seeing to it that it's opposite holds.
Further elements such as (missing) knowledge \cite{broersen2011epistemic}, intention \cite{halpern2018} or other excuses \cite{fischer2011} may be added or left out, depending on the desired complexity of the representation.

%%% Our work

In contrast to the often unquestioned assumption of actual causation as a precondition for responsibility there are scenarios in which there is no actual causality but we would still want to speak of responsibility. For example, in legal contexts `attempts' are well-known cases of non-actualized causal relations \cite{moore2019}. 
\cite{baltag2020} try to solve this by retroactively equipping stit logics with a notion of \emph{action failure} or \emph{attempt}. 
More comprehensively, in some of the occasions in which we want to assign responsibility but cannot fall back on actual causation it has been observed that the action in question \emph{increased the probability} of the outcome \cite{broersen2011attempt}.
Note, however, that we can have both causation without probability increase -- due to an initial lowering of the probability of unlikely causes -- and probability increase without causation -- due to fizzling or probabilistic preemption.\footnote{Fizzling here describes an event in which the probability of the outcome was increased but stays below 1 and the unwanted outcome does not occur due to chance. Probabilistic preemption is simply preemption in a probabilistic context, that is, the event which increased the probability of the outcome does not turn out to be a cause as a different event interrupts the causal chain by causing the outcome first.} Thus, attempts to reduce causation to probability raising, as described by \cite{hitchcock2021,kvart2002,vallentyne2008}, do not seem very promising.

We argue that causation is not probability increase, but that both are separate grounds for responsibility, as was already suggested by \cite{moore_causing_2007}: ``These metaphysical conclusions are enough to make possible the moral thesis I now want to advance: that chance-raising is an independent desert basis [for responsibility], along with causation and counterfactual dependence.''\footnote{The well worked out reduction by Halpern and Pearl of causation to counterfactual dependence \emph{in an appropriately modified context} may serve as an argument to disregard the third disjunct \cite{halpern_causes_2005}.}

%\paragraph{Our Approach.}
As responsibility due to actual causation has been the focus in the literature so far, we now analyse the aspect of responsibility due to probability increase.
In a previous work (that is currently under review) a subset of the authors suggested some functions for computing responsibility based on probability increase and evaluated these according to a number of paradigmatic example scenarios. 
The current work provides a more formalised and in-depth examination of the properties of responsibility functions.
We decompose the responsibility computation for a group of agents into two parts: First we compute values for individual decisions by group members leading to the outcome, before aggregating these values into the final group responsibility. 

In order to evaluate proposed responsibility functions we employ an axiomatic method \cite {thomson2001}. That is, we formulate desirable properties which our responsibility functions are to have in the form of axioms, and assess the quality of (sets of) proposed functions according to these axioms. Additionally, we will present certain incompatibility and implication results between the axioms and a characterisation of one specific subset of functions\kick{ using some of the proposed axioms}. Two responsibility functions will be singled out to be maximally axiom-compliant with respect to our impossibility result, one respecting the lower-bound axioms and one respecting the upper-bound ones.

\section{Framework and Responsibility Functions}
\label{sec:framework}

We are interested in modelling decision situations of several agents interacting over time. Additionally, we want to allow for the presence of various forms of uncertainty, resulting from the interactive nature of the scenario or from an inherent uncertainty (both probabilistic or not) regarding the outcomes of actions.
As noted earlier, the framework presented here also occurs in a different work currently under review which is authored by a subset of the authors.

\subsection{Framework}\label{sec:framework-framework}

In order to represent decision scenarios we use extensive-form game trees equipped with specific additions for the features mentioned above. Formally, the framework is the following.

\begin{definition}[\textbf{Decision tree}] A morally evaluated multi-agent \textbf{decision tree} with uncertainty (without possible confusion in the remainder abbreviated as {\em tree}) is a tuple \\
\[\T = \langle I, V, E,\sim,(A_v),(c_v),(p_v), \epsilon \rangle\]
where:
\begin{itemize}
\item  $I$ is a nonempty finite set of {\em agents} (or players).
\item $\langle V,E\rangle$ is a directed rooted tree with nodes $V$ and edges $E$. $V = \bigcup_{i\in I} V_i \cup V_a \cup V_p\cup V_o$ is partitioned into a set of {\em decision nodes} $V_i$ for each agent $i\in I$, 
a set of {\em ambiguity nodes} $V_a$, 
a set of {\em probabilistic uncertainty nodes} $V_p$, 
and a set of {\em outcome nodes} $V_o$, that are exactly the leaves of the tree.
      We denote the set of all decision nodes by
      $V_d := \bigcup_{i\in I} V_i $.

\item $\sim$ is an equivalence relation on $V_d$ so that $v\sim w$ implies $v, w \in V_i$ for the same agent $i \in I$. 
      We call $\sim$ {\em information equivalence} and the equivalence classes of $\sim$ in $V_i$ the {\em information sets} of $i$. 
\item $A_v$ is a nonempty, finite set of $i$'s possible {\em actions} in $v$, for each agent $i\in I$ and decision node $v\in V_i$. Whenever $v \sim w$, $A_{v} = A_{w}$. Let $\bigcup \limits_{v \in V_d} A_v =: \mathcal A$.

\item $c_v:A_v\to S_v$, where $S_v := \{ w\in V: (v,w)\in E \}$, is a bijective {\em consequence function} mapping actions to successor nodes for each node $v\in V\setminus V_o$.

\item $p_v\in\Delta(S_v)$ for $v\in V_p$, with $\Delta (A)$ being the set of all probability distributions on a given set $A$, is a probability distribution on the set of possible successor nodes for each probability node. 

\item $\epsilon \subseteq V_o$ is a set of ethically undesirable outcomes. 

\end{itemize}

\end{definition}

The set of {\em agents} and the {\em directed tree} encode a multi-agent decision situation with the direction of the tree showing the temporal progression. {\em Decision, ambiguity} and {\em probability nodes} receive their intuitive interpretation. Information equivalence $\sim$ encodes which nodes are indistinguishable to an agent at a given time. Note that only decision nodes of one agent may be related -- that is, i) agents always know whether or not it is their turn, and ii) we do not need to index the equivalence relation with the agent under consideration. We use the term {\em uncertainty} as an umbrella term for information uncertainty, ambiguity and probabilistic uncertainty.\footnote{Note that according to this representation we do not assign probabilities to any player's actions. While it can very well be argued that people do indeed reason using subjective probability distributions over other players' actions, we do not include this here for several reasons. The first one is that it is extremely difficult, if not impossible, to assess these probabilities objectively. The other reason is that including reasoning via other agents' expected actions leads to possible excuses in responsibility ascription that we do not want to permit.}

\paragraph{Graphical representation.} We represent these features graphically as follows. An example of two simple scenarios is depicted in \Cref{fig:scenarios}.

    Directed edges are shown as arrows labelled with actions, descriptions of the situation, or probabilities, depending on the preceding node. Decision nodes are depicted as diamonds labelled with the respective agent and ambiguity nodes are empty diamonds, in line with the intuition that these can be regarded as a decision node of a special agent that represents the decision scenario's environment.  Probabilistic uncertainty nodes are depicted as empty squares. Outcome nodes are represented as circles, with undesirable ones shaded in gray. Information sets are connected via a dotted line.

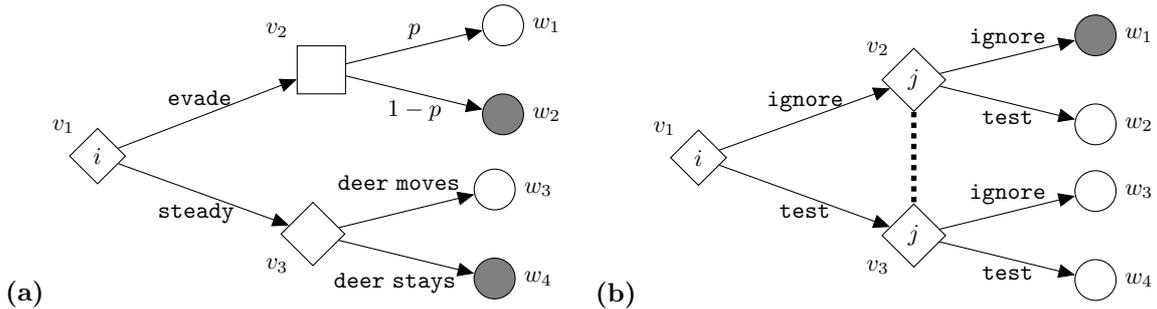
\begin{figure}
\centering
{
    {\bf (a)}\begin{tikzpicture} [->, scale=.9, every node/.style={scale=.9}, node distance=1cm, decision node/.style={diamond, draw, aspect=1, minimum height= 7mm},  prob node/.style={regular polygon,regular polygon sides=4, draw, minimum height= 1cm},  dummy/.style = {}, good outcome/.style={circle, draw, minimum width=6mm}, bad outcome/.style={circle, draw, fill=gray, minimum width=6mm }, >= triangle 45]
\node[decision node] (1) [] [label=above left: $v_1$] {$i$};

\node[dummy] (a) [right of = 1] [] {};
\node[dummy] (a2) [right of = a] [] {};

\node[prob node] (2) [above right =  of a2] [label=above left: $v_2$] { };
\node[decision node] (3) [below right = of a2] [label=below left: $v_3$] { \ \ \ \ \ };

\node[dummy] (b) [right of = 2] [] {};
\node[dummy] (b2) [right of = b] [] {};
\node[dummy] (b3) [right of = b2] [] {};

\node[good outcome] (good2) [above right = .4cm of b2] [label = right: $w_1$]{}; %label = right: $v_1$
\node[bad outcome] (bad2) [below right = .4cm of b2] [label = right: $w_2$] {};

\node[dummy] (c) [right of = 3] [] {};
\node[dummy] (c2) [right of = c] [] {};

\node[good outcome] (good3) [above right = .4cm of c2] [label = right: $w_3$] {};
\node[bad outcome] (bad3) [below right = .4cm of c2] [label = right: $w_4$] {};

\path (1) edge node[above] {$\mathtt{evade \ \ }$} (2) (1) edge node[below] {$\mathtt{steady \ \ }$} (3) (2) edge node[above] {$p$}  (good2) (2) edge node[below] {$1 - p$} (bad2) (3) edge node[above=5pt] {$\mathtt{deer \ moves\ \ }$} (good3) (3) edge node[below=3pt] {$\mathtt{deer \ stays\ \ \ }$} (bad3) ;

\end{tikzpicture}
}
{
    {\bf (b)}
    \begin{tikzpicture} [-, scale=.9, every node/.style={scale=.9}, node distance=1cm, decision node/.style={diamond, draw, aspect=1, minimum height= 7mm},  prob node/.style={regular polygon,regular polygon sides=4, draw, minimum height= 1cm},  dummy/.style = {}, good outcome/.style={circle, draw, minimum width=6mm}, bad outcome/.style={circle, draw, fill=gray, minimum width=6mm }, >= triangle 45]
\node[decision node] (1) [] [label=above left: $v_1$] {$i$};

\node[dummy] (a) [right of = 1] [] {};
\node[dummy] (a2) [right of = a] [] {};

\node[decision node] (2) [above right =  of a2] [label=above left: $v_2$] {$j$}; %
\node[decision node] (3) [below right = of a2] [label=below left: $v_3$] {$j$};

\node[dummy] (b) [right of = 2] [] {};
\node[dummy] (b2) [right of = b] [] {};
\node[dummy] (b3) [right of = b2] [] {};

\node[bad outcome] (good2) [above right = .4cm of b2] [label = right: $w_1$]{};
\node[good outcome] (bad2) [below right = .4cm of b2] [label = right: $w_2$] {};

\node[dummy] (c) [right of = 3] [] {};
\node[dummy] (c2) [right of = c] [] {};

\node[good outcome] (bad3) [above right = .4cm of c2] [label = right: $w_3$] {};
\node[good outcome] (good3) [below right = .4cm of c2] [label = right: $w_4$] {};

\path[->] (1) edge node[above] {$\mathtt{ignore}$}  (2) (1) edge node[below] {$\mathtt{test}$} (3) (2) edge node[above] {$\mathtt{ignore}$}  (good2) (2) edge node[below] {$\mathtt{test}$} (bad2) (3) edge node[below] {$\mathtt{test}$} (good3) (3) edge node[above] {$\mathtt{ignore}$} (bad3);

\draw[line width=2pt, dotted=on] (2) -- (3); 

\end{tikzpicture}}

\caption{\label{fig:scenarios} Graphical representation of interactive decision scenarios with uncertainty in the described framework. (a) An autonomous vehicle (agent $i$) approaches a deer on a street. It can either try to evade the deer with a computed success rate of $p$ or keep steady, resulting in a crash if the deer does not move.
(b) Agents $i$ and $j$ are two instances in a distributed system which may perform a costly test for the failure of a component, with $i$ coming before $j$. Both may either perform the test or ignore it, but $j$ has no way of knowing whether $i$ performed the test. If they both ignore the test a failure is overlooked.
}
\end{figure}

\paragraph{Further notions.} Additionally, we will need formal representations of the following notions defined on this framework: that of \emph{past} and \emph{future} of a certain node, that of \emph{strategies}, which tell us for a given set of agents $G$ all future decisions of $G$, and that of \emph{scenarios}, which tell us about others' future decisions and resolve all non-probabilistic uncertainty. Formally, these are defined as follows.

\begin{definition}[\textbf{History, Branch, Information Branch}]
Let a tree $\T$ and a node $v\in V$ be given. Then the \textbf{history} of $v$ is 
\[ H(v):= \{ v' \in V \mid \exists v_1,\ldots,v_n \in V \text{ s.t. } (v',v_1), (v_1, v_2), \ldots, (v_n, v) \in E \}, \]
and the \textbf{branch} of $v$ is 
\[ B(v):= \{ v' \in V \mid v \in H(v')\}, \] 
that is, all possible futures after $v$. The \textbf{information branch} of $v$ is $B^{\sim}(v):= \bigcup_{v'\sim v} B(v') $, that is, all futures considered possible in the information set of $v$.
\end{definition} 

Given a group $G$, we call their decision nodes $V_G:=\bigcup_{i\in G} V_i $, and the decision nodes of all agents not in the group $I\setminus G$ together with all ambiguity nodes $V_{-G}:= \bigcup_{i\in I\setminus G} V_i \cup V_a $. With these two notations we can define the terms \emph{strategy} and \emph{scenario}.

\begin{definition}[\textbf{Strategy}]
Let a tree $\T$, a group $G\subseteq I$, and a node $v\in V$ be given. Then a \textbf{strategy} for $G$ at $v$ is a function \[ \sigma \colon V_G \cap B^{\sim}(v) \to \bigcup_{v_d \in V_G\cap B^{\sim}(v)} A_{v_d} \]
that chooses an action for every one of $G$'s future decision nodes considered possible according to the information set in $v$, i.e., $\sigma(v) \in A_v$. Additionally, we require that the actions selected are not just logically but also practically possible, that is, a strategy must select the same action for all nodes in one information set: $ \sigma (v_d) = \sigma (v'_d) $ for all $v_d, v'_d \in V_G \cap B^{\sim} (v) \text{ s.t. } v_d \sim v'_d $.\footnote{This is sometimes called a \emph{uniform strategy} \cite{yazdanpanah_strategic_2019}, however, as this is the only kind of strategy we use here we leave out the qualification. Also, our restriction is only concerned with individual knowledge at the agent level and does not concern any kind of group knowledge or coordination. Therefore, we can employ this restriction without going into the debate about what actions a group can collectively select.}  We denote the set of all these functions by $\Sigma(\T, G, v)$.
\end{definition}

\begin{definition}[\textbf{Scenario}]
Again, let a tree $\T$, a group $G \subseteq I$, and a node $v\in V_G$ be given. Then a \textbf{scenario} for $G$ at $v$ is a node $v_{\zeta} \sim v$ together with a function 
\[ \zeta \colon V_{-G} \cap B(v_{\zeta}) \to \bigcup_{v'\in V_{-G} \cap B(v_{\zeta})} S_{v'} \]
that chooses successor nodes $\zeta (v')$ for all future decision nodes by non-members of G or ambiguity nodes, i.e., $\zeta(v) \in S_{v}$. The node $v_{\zeta}$ is considered the actual node in scenario $\zeta$, thus resolving information uncertainty. The function $\zeta$ resolves non-probabilistic uncertainty regarding what will happen in the future (disregarding $G$'s own actions).

We denote the set of all these functions by  $Z^\sim(\T, G, v)$. Note that the specification of $v_{\zeta}$ is implicit in the definition of $\zeta$.
\end{definition}

Given a node $v\in V_G$, a strategy $\sigma\in\Sigma(\T, G, v)$, and a scenario $\zeta\in Z^\sim(\T, G, v)$, one can easily compute in a recursive way the conditional probability to end in an outcome node $w$ given that the current node is actually $v_\zeta$ and the choices in non-probabilistic nodes are made according to what $\sigma$ and $\zeta$ specify. Let's call this conditional probability the {\em likelihood}~$\ell(w | v, \sigma, \zeta)$. We are often interested in the likelihood of ending in any undesirable outcome node $w \in \epsilon$ for which we slightly abuse notation and write $\ell(\epsilon | v, \sigma, \zeta) \coloneqq \sum_{w \in \epsilon} \ell(w | v, \sigma, \zeta)$.

In the remainder of the paper, unless specified otherwise, we will use $v, v_1, v_2, \ldots$ to indicate decision nodes and $w, w_1, w_2, \ldots$ to indicate outcome nodes.

\subsection{Responsibility Functions}\label{sec:framework-respfunctions}

We are interested in assigning backward looking moral responsibility to a (possibly singleton) group for a certain outcome. In order to compute this, we utilize the path from the root node of the decision tree to the given outcome node. 
\kick{That is in particular the information which decisions the agents in the group took in each of the decision nodes along the way.} We will assign an intermediate value, which we call ``member contribution'', for single decisions by individual group members in a given node. Consequently, we use an aggregation function to combine these values into the final ``outcome responsibility''.\footnote{By outcome responsibility we mean moral responsibility assigned at an outcome node as opposed to at an intermediate node, and not non-moral causal responsibility (as it is sometimes understood in the literature, e.g. \cite{vincent2011}).} 

Member contribution is not computed irrespective of an agent's group membership but precisely with this information in mind. Thus, the point of view we take here does not stand in contrast to non-reducibility results known from the literature \cite{braham2018,duijf2018,tamminga2020}.

It is important to note that we place no restrictions on which agents can form a group. This is a decision the modeler takes at the time of the computation of responsibility. Therefore, we also do not place any \emph{a priori} restrictions on the level of coordination within an arbitrary group. Finally, we always assume common knowledge of the full decision situation. 

Member contribution is a function of the individual agent and the group that they are a member of, the node in which the current decision is taken as well as the action which the agent selects. All of this is considered within a given decision situation. We measure contribution as an increase in probability. Therefore the range of this function will be the interval $[0,1]$.\footnote{Note that a member contribution function $\r$ only receives information internal to the scenario modeled in the decision tree. The ascribed contribution is thus invariant under any outside information. This includes invariance with respect to the specific name an agent is given as well as any form of duty or other historic information.}

\begin{definition}[\textbf{Member Contribution}] Let a tree $\T$ be given. 
\textbf{Member contribution} is a function mapping a set (`group') of agents $G\subseteq I$, an agent $i\in G$, a node $v\in V_i$ and an action $\a\in A_v$ to values in the interval $[0,1]$: \[ \r \colon \mathcal P (I) \times I \times V_d \times \mathcal A \to [0,1].\]
Whenever $i \notin G$, $v\notin V_i$, or $\a\notin A_v$ the function is not defined.
\end{definition}

Next we define the general form of the functions we use to aggregate member contribution into outcome responsibility (see  \cite{mesiar2008} for a general introduction to aggregation functions).

\begin{definition}[\textbf{Aggregation function}] \
%Let $I$ be a subinterval of the extended real line, $I = [a,b] \subseteq [-\infty, \infty]$.
For $n\geq1$ an {\em n-ary aggregation function} is a function $\agg^{(n)}: [0,1]^n \to [0,1]$ that is non decreasing in each entry and fulfills the following boundary conditions 
\[\inf_{\mathbf{x}\in [0,1]^n} \agg^{(n)}(\mathbf{x} ) = \inf [0,1] = 0 \quad \text{ and } \quad \sup_{\mathbf{x}\in [0,1]^n} \agg^{(n)}(\mathbf{x} ) = \sup[0,1] = 1.\]
An {\em extended aggregation function} is a function $\agg: \bigcup \limits_{n\in \mathbb N_{\geq 1}} [0,1]^{n} \to [0,1]$ such that for all $n>1$, $\agg^{(n)} = \agg \restriction _{[0,1]^n}$ is an n-ary aggregation function and $\agg^{(1)}$ is the identity on $[0,1]$.
\end{definition}

Finally, outcome responsibility will be a function of a group of agents and an outcome node, again in a specified decision situation. \kick{Some of the functions suggested below will not be aggregation functions in the sense defined above, by virtue of their image not being restricted to the interval $[0,1] $. We will allow for real values as outcome responsibility in order to be able to talk about these functions nevertheless.}

\begin{definition}[\textbf{Outcome responsibility}] Let a tree $\T$ be given. \textbf{Outcome responsibility} is a function of a group of agents  $G\subseteq I$, and an outcome node $w\in V_o$ to the reals: \[ \R \colon \mathcal P (I) \times V_o \to \mathbb R.\] 
We consider outcome responsibility for group $G \subseteq I$ at node $w$ to be a function of the member contribution values of the individual decisions of agents $i\in G$ which lead to the outcome $w$: 
\[\R(G, w) := \agg\left((\r(G, i, v, \mathtt{a}_{v \to w} ))_{i\in G, v \in V_i\cap H(w)}\right),\]
for an appropriate aggregation function $\agg$.
$\a_{v \to w}$ is that action available in $v$ for which the successor node lies on the path to $w$.
\end{definition}

In the following three sections we provide an axiomatic study of member contribution, aggregation functions and outcome responsibility.

\section{Member Contribution Functions}\label{sec:pwresp}

We start our analysis by considering member contribution functions in more detail. After establishing a list of interesting axioms that reflect certain moral desiderata we present three functions capturing different notions of contribution. We conclude by qualitatively evaluating the presented functions with respect to the axioms.

\subsection{Axioms for Member Contribution Functions}\label{sec:pwresp-axioms}

The axioms described here demand certain contribution values in clearly defined situations. This reflects some moral desiderata that might (or might not) be important in different situations.
The first axiom reflects the idea that in situations the agent can not discern between they should not be treated differently.

\begin{description}
\item[(KSym)]{\em Knowledge Symmetry.}
  The same action taken in the same information set produces the same contribution value. Let $v\sim v' $ and $\a\in A_{v} = A_{v'}$. Then $\r(G, i, v, \a) = \r(G, i, v', \a)$.
\end{description}

The following axioms bound the assigned contribution value either to 0 or to 1 for taking specific actions. For both pairs of axioms it is easy to see that the version of the axiom considering all information branches (marked by the ${}^\sim$) is logically implied by the original one. 

\begin{description}
\item[(AMC)]{\em Avoidance of member contribution.} 
  Consider a tree $\T$, a group $G \subseteq I$ and a node $v\in V_i$ for some $i\in G$ with $A_{v} = \{\a, \b_1, \ldots, \b_m \},$ where $m\geq 0,$ such that choosing $\a$ certainly leads to a desirable outcome\kick{, $B(c_{v}(\a)) \cap V_o \subseteq \epsilon^c$,} and choosing any other option certainly leads to an undesirable outcome\kick{, $ ( (B(c_{v}(\b_1)) \cup \ldots \cup B(c_{v}(\b_m))) \cap V_o \subseteq \epsilon$}. Then $\r (G, i, v, \a) = 0 $.

\item[(AMC$^\sim$)]{\em Avoidance of member contribution in information branches.}
  Consider a tree $\T$, a group $G \subseteq I$ and a node $v\in V_i$ for some $i\in G$ with $A_{v} = \{\a, \b_1, \ldots, \b_m \}, $ where $m\geq 0,$ such that choosing $\a$ in any $v' \sim v$ certainly leads to a desirable outcome\kick{, $B^\sim (c_{v}(a)) \cap V_o \subseteq \epsilon^c$,} and choosing any other option certainly leads to an undesirable outcome\kick{, $ ( (B^\sim (c_{v}(\b_1)) \cup \ldots \cup B^\sim (c_{v}(\b_m))) \cap V_o \subseteq \epsilon$}. Then $\r (G, i, v, \a) = 0 $.

\item[(FMC)]{\em Full member contribution.}
  Consider a tree $\T$, a group $G \subseteq I$ and a node $v\in V_i$ for some $i\in G$ with $A_{v} = \{\a, \b_1, \ldots, \b_m \},$ where $ m > 0,$\footnote{Note that it is important here to enforce that there is actually a choice for the agent, i.e., that $m \neq 0$.} such that choosing $\a$ certainly leads to an undesirable outcome\kick{, $B(c_{v}(\a)) \cap V_o \subseteq \epsilon$,} and choosing any other option certainly leads to a desirable outcome\kick{, $ ( (B(c_{v}(\b_1)) \cup \ldots \cup B(c_{v}(\b_m))) \cap V_o \subseteq \epsilon^c$}. Then $\r (G, i, v, \a) = 1 $.

\item[(FMC$^\sim$)]{\em Full member contribution in information branches.}
  Consider a tree $\T$, a group $G \subseteq I$ and a node $v\in V_i$ for some $i\in G$ with $A_{v} = \{\a, \b_1, \ldots, \b_m \},$ where $ m > 0$, such that choosing $\a$ in any $v' \sim v$ certainly leads to an undesirable outcome\kick{, $B^\sim (c_{v}(\a)) \cap V_o \subseteq \epsilon$,} and choosing any other option certainly leads to a desirable outcome\kick{, $ ( (B^\sim (c_{v}(\b_1)) \cup \ldots \cup B^\sim (c_{v}(\b_m))) \cap V_o \subseteq \epsilon^c$}. Then $\r (G, i, v, \a) = 1 $.
\end{description}

The axioms (AMC) and (FMC) represent notions of {\em seeing to it that} \cite{stit2001}.
They request that if an agent ensures a desirable or undesirable outcome through a single action in one decision node, i.e. she {\em sees to it that} this outcome holds, she is to be assigned no or full contribution value, respectively.

\subsection{Proposed Functions}\label{sec:pwresp-functions}

We consider three different member contribution functions. The first
follows an ad-hoc way of measuring responsibility through increased likelihood of an undesirable outcome. The other two use more elaborate notions of responsibility through risk taking and through negligence. Let a tree $\T$, group $G\subseteq I$, agent $i\in G$ and node $v\in V_i$ be given.

\paragraph{Contribution through increase in guaranteed likelihood.}
Here, agent $i$ is assigned contribution value for the undesirable outcome if, and to the degree that, their action increased the \emph{guaranteed likelihood} of $\epsilon$. 
This is often the first intuition when adopting a probabilistic view of causation in responsibility ascription, see e.g. \cite{kaiserman_more_2018}, or \cite{vallentyne2008}: ``I shall assume that the relevant causal connection is that the choice increases the objective chance that the outcome will occur''.
We define the \emph{known guaranteed likelihood} of $\epsilon$ at node $v \in V$ as 
\[ \gamma(v) \coloneqq \min_{\sigma \in \Sigma(\T, G, v)} \min_{\zeta \in Z^\sim(\T,G,v)} \ell(\epsilon | v,\sigma, \zeta). \]

The member contribution we assign to agent $i \in G$ for performing action $\a \in A_v$ in node $v\in V_i$ is measured as the increase in guaranteed likelihood as follows
\[ \rlike(G, i, v,\a) \coloneqq \Delta \gamma(v, \a) \coloneqq \gamma(c_v(\a)) - \gamma(v).\]

\paragraph{Contribution through risk taking.}
This variant assigns a contribution value to an agent if their action can be seen as some kind of risk taking. We define taking a risk here as not avoiding a possible bad outcome.
We define the \emph{optimal avoidance} of $\epsilon$ by group $G$ given a scenario $\zeta\in Z^\sim(v)$ as 
\[ \omega(v,\zeta) \coloneqq  \min_{\sigma\in\Sigma(\T, G, v)} \ell(\epsilon | v,\sigma,\zeta). \]

\noindent We now measure member contribution as the shortfall in avoiding $\epsilon$ due to action $\a$ at $v$
\[ \rrisk(G,i, v, \a) \coloneqq \max_{\zeta \in Z^\sim(\T, G, v)} \Delta \omega(v, \zeta, \a) \coloneqq \max_{\zeta \in Z^\sim(\T, G, v)} [\omega (c_{v}(\a), \zeta) - \omega(v, \zeta)]. \]
That is, rather than assuming a single scenario we compare $G$'s strategies over all scenarios.

\paragraph{Contribution through negligence.}
The third and final variant we study in this paper assigns the contribution value in a similar way to the previous variant but deducts a baseline of unavoidable risk. This ensures that in situations where all available actions produce some risk of leading to an undesirable outcome the agent is only assigned a contribution value when choosing negligently. 
We define the minimal risk formally through
\[\underline\rho(G, i, v) \coloneqq \min_{\a\in A_{v}} \rho(G, i, v,\a) = \min_{\a\in A_{v}} \max_{\zeta \in Z^\sim(\T, G, v)} \Delta \omega(v, \zeta, \a) = \min_{\a \in A_{v}} \rrisk(G, i, v, \a) . \]
With that we define this variant of member contribution as 
\[ \rneg(G, i, v, \a) \coloneqq \Delta \rho(G, i, v, \a) \coloneqq \max_{\zeta \in Z^\sim(\T, G, v)} \Delta \omega(v, \zeta, \a) - \underline\rho(G, i, v). \]

\subsection{Application to Example Scenarios}\label{sec:pwresp-examples}

The suggested functions give the following results for the examples from \Cref{fig:scenarios} and \Cref{fig:coordination_rescue}. The computations can be found in \Cref{app:pw_examples}. We can see that in many situations $\rlike$ is not very fine-grained. Contribution values for members of the group $\{i,j\}$ are often zero as we can assume that the other group members will follow a strategy that avoids the undesirable outcome. Also, we see that in accordance with our intentions $\rrisk$ may assign contribution values where $\rneg$ fails to do so.

\begin{table}
 \centering
 \small
  \begin{tabular}{C{1cm}|C{1.2cm}C{1.2cm}|C{1.5cm}C{1.5cm}|C{1.2cm}C{1.2cm}|C{1.5cm}C{1.5cm}|C{1.2cm}C{1.2cm}}
  \toprule
     &
      \multicolumn{2}{c|}{\textbf{\Cref{fig:scenarios}(a)}} &
      \multicolumn{4}{c|}{\textbf{\Cref{fig:scenarios}(b)}} &
      \multicolumn{4}{c}{\textbf{\Cref{fig:coordination_rescue}}} \\
     & \multicolumn{2}{c|}{$G=\{i\}$} 
     & \multicolumn{2}{c|}{$G=\{i\}$ or $G=\{ j\}$} 
     & \multicolumn{2}{c|}{$G=\{i,j\}$} 
     & \multicolumn{2}{c|}{$G=\{i\}$ or $G=\{ j\}$} 
     & \multicolumn{2}{c}{$G=\{i,j\}$} \\
    & $\mathtt{evade}$ & $\mathtt{steady}$ & $\mathtt{ignore}$ & $\mathtt{test}$ & $\mathtt{ignore}$ & $\mathtt{test}$ & $\mathtt{left}$ & $\mathtt{right}$ & $\mathtt{left }$ & $\mathtt{right}$\\
    \midrule
    $\rlike$ & 1-p & 0 & 0 & 0 & 0 & 0 & 0 & 0 & 0 & 0\\
   \rowcolor{lightgray}  $\rrisk$ & 1-p & p & 1 & 0 & 0 & 0 & 1 & 1 & 0 & 0\\
    $\rneg$ & 1-2p & 0 & 1 & 0 & 0 & 0 & 0 & 0 & 0 & 0\\
    \bottomrule
  \end{tabular}
  \caption{\label{tbl:pw-examples}Evaluation of member contribution values with $\T$ as in the examples from \Cref{fig:scenarios} and \Cref{fig:coordination_rescue}. $G$ and $\a$ are as specified, $v$ is clear from the context. In cases where $G=\{i,j\}$ the values for agents $i$ and $j$ are the same. We assume $p< 1-p$.}
\end{table}

\subsection{Evaluation of Proposed Functions}\label{sec:pwresp-evaluation}

In this section we check compliance of the member contribution functions described earlier with our axioms. The proof of the following proposition can be found in \Cref{app:pw_proofs}. Note that the coordination game as depicted in \Cref{fig:coordination_rescue} works as a counterexample for all negative results.

\begin{table}
\centering
\small
\begin{tabular}{cccccc} \toprule
                        \phantom{\textbf{Variant}}  & (KSym)        & (AMC)         & (AMC$^\sim$)  & (FMC)         & (FMC$^\sim$)\\ \midrule
                        % $\rstrict$        & $\times$      & \checkmark    & \checkmark     & \checkmark   & \checkmark \\
                        $\rlike$          & $\times$      & \checkmark    & \checkmark     & \checkmark   & \checkmark \\
\rowcolor{lightgray}    $\rrisk$         & \checkmark    & $\times$      & \checkmark     & \checkmark   & \checkmark \\
                        $\rneg$           & \checkmark    & \checkmark    & \checkmark     & $\times$     & \checkmark \\ \bottomrule
\end{tabular}
\caption{\label{tbl:axioms_pw_1}
Summary of the compliance by the three proposed functions with the considered axioms. ``\checkmark'' denotes that the function satisfies the axiom whereas ``$\times$'' denotes that there are instances where the axiom is not satisfied.
}
\end{table}

\begin{restatable}{proposition}{pwrespaxioms}\label{prop:pw_resp_axioms}
The three functions $\rlike, \rrisk$ and $\rneg$ satisfy the given axioms as specified in \Cref{tbl:axioms_pw_1}.
\end{restatable}

We now show that one has to choose between assigning unavoidable contribution values in some situations -- and thus violating an axiom like (AMC) -- or having situations where we cannot assign full contribution values -- and in turn violating an axiom like (FMC).
We will come back to this point later for similar axioms regarding aggregated outcome responsibility.

\begin{restatable}{proposition}{AMCvsFMC}\label{thm:AMCvsFMC}
No function can simultaneously satisfy (KSym), (AMC), and (FMC).
\end{restatable}
\begin{proof}
Consider the coordination game depicted in \Cref{fig:coordination_rescue} and assume that the member contribution function $\r$ satisfies (KSym), (AMC) and (FMC). By (AMC), $\r(j, v_1, \mathtt{cinema}) = \r(j, v_2, \mathtt{theater}) = 0$ and by (FMC) we have $\r(j, v_1, \mathtt{theater}) = \r(j, v_2, \mathtt{cinema}) = 1$. But this contradicts (KSym) as $\r(j, v_1, \mathtt{cinema}) \neq \r(j, v_2, \mathtt{cinema})$.\qed
\end{proof}

\begin{figure}
\centering
% \subfigure{
    \begin{tikzpicture} [-, scale=.9, every node/.style={scale=.9}, node distance=1cm, decision node/.style={diamond, draw, aspect=1, minimum height= 7mm},  prob node/.style={regular polygon,regular polygon sides=4, draw, minimum height= 1cm},  dummy/.style = {}, good outcome/.style={circle, draw, minimum width=6mm}, bad outcome/.style={circle, draw, fill=gray, minimum width=6mm }, >= triangle 45]
\node[decision node] (1) [] [label=above left: $v_1$] {$i$};

\node[dummy] (a) [right of = 1] [] {};
\node[dummy] (a1) [right of = a] [] {};
\node[dummy] (a2) [right of = a1] [] {};

\node[decision node] (2) [above right =  of a2] [label=above left: $v_2$] {$j$};
\node[decision node] (3) [below right = of a2] [label=below left: $v_3$] {$j$};

\node[dummy] (b) [right of = 2] [] {};
\node[dummy] (b1) [right of = b] [] {};
\node[dummy] (b2) [right of = b1] [] {};
\node[dummy] (b3) [right of = b2] [] {};

\node[good outcome] (good2) [above right = .4cm of b2] [label = right: $w_1$]{};
\node[bad outcome] (bad2) [below right = .4cm of b2] [label = right: $w_2$] {};

\node[dummy] (c) [right of = 3] [] {};
\node[dummy] (c1) [right of = c] [] {};
\node[dummy] (c2) [right of = c1] [] {};

\node[bad outcome] (bad3) [above right = .4cm of c2] [label = right: $w_3$] {};
\node[good outcome] (good3) [below right = .4cm of c2] [label = right: $w_4$] {};

\path[->] (1) edge node[above] {$\mathtt{left}$}  (2) (1) edge node[below] {$\mathtt{right}$} (3) (2) edge node[above] {$\mathtt{left}$}  (good2) (2) edge node[below] {$\mathtt{right}$} (bad2) (3) edge node[below] {$\mathtt{right}$} (good3) (3) edge node[above] {$\mathtt{left}$} (bad3);

\draw[line width=2pt, dotted=on] (2) -- (3); 

\end{tikzpicture}

\caption{\label{fig:coordination_rescue} A coordination game where two autonomous vehicles (agents $i$ and $j$) approach each other on a street with two lanes, destined to crash into each other if they don't evade. The crash is avoided if and only if both vehicles move to their respective right or both move to their respective left.}
\end{figure}
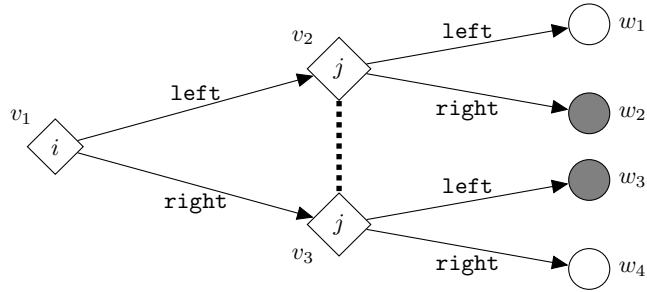

On the other hand, we can see in \Cref{tbl:axioms_pw_1} that  each combination of two out of the three axioms (AMC), (FMC), and (KSym) is possible.
We would argue that sacrificing (KSym) might not be satisfactory in most instances, as responsibility assignment should not depend on factors outside the agents' knowledge.
Instead, one could require (KSym) together with the corresponding weakened variants (AMC$^\sim$) and (FMC$^\sim$).

\section{Aggregation Functions}\label{sec:agg}

After studying member contribution functions and their axiomatic compliance, we will now turn to aggregation functions. Again, we start by defining a number of axioms applicable to these, before presenting a set of possible functions that could be used as aggregation functions and qualitatively compare them using the axioms.

\subsection{Axioms for Aggregation Functions}\label{sec:agg-axioms}

As a first axiom we consider monotonicity. If we input more entries of non-zero member contribution into the computation, we will want the outcome responsibility to increase as well. Similarly, if we input the same number of member contributions but some of them have larger values, we also expect the outcome responsibility to be larger (instead of just larger or equal, as is demanded in the definition of aggregation functions).

\begin{description}
    \item[(BSM$^+$)] {\em Bounded strict monotonicity under adding non-zero entries.} For $[x_1, \ldots, x_n] \in [0,1]^n$ and $x_{n+1} \in (0,1]$ it holds that \[\text{ either }\quad \agg([x_1, \ldots, x_n]) = 1 \quad \text{ or }\quad \agg([x_1, \ldots, x_n]) < \agg([x_1, \ldots, x_n, x_{n+1}]).\]
    \item[(BSM$^>$)] {\em Bounded strict monotonicity under increasing entries.} For $[x_1, \ldots, x_n], [y_1, \ldots, y_n] \in [0,1]^n$ where there is a  $j \in \{1,\dots,n\}$ s.t. $x_j < y_j$ and $x_i = y_i$ for all  $i \neq j$ it holds that
        \[\text{ either }\quad \agg([x_1, \ldots, x_n]) = 1 \quad \text{ or }\quad \agg([x_1, \ldots, x_n]) < \agg([y_1, \ldots, y_n]).\]
    \item[(BSM)] {\em Bounded strict monotonicity.} Both of the above hold.
\end{description}

Strengthening this idea, we can ask for the aggregation function to not only be monotone, but linear in each component. That is, if we keep the other input values fixed, the influence of one member contribution on the outcome should be linear.
In order to define this we use the following notational convention.
For an n-ary aggregation function $\agg^{(n)}$, an n-dimensional vector $\overline{\r} = (\r_1, \ldots, \r_n)$, and an index $i\leq n$, let $\agg^{(n)}_{\overline{\r},-i} (x) := \agg^{(n)} (\r_1, \ldots, \r_{i-1}, x, \r_{i+1},\ldots, \r_n)$ be the function obtained from $\agg^{(n)}$ by keeping all input values fixed as given by $\overline{\r}$, except for the i-th entry.
% 
% Now the axiom can be defined as follows.
\begin{description}
    \item[(LIN)] {\em Linearity in each component.} For each $n\in \mathbb N_{\geq 1}$, for each input vector of length $n$,  $\overline{\r} \in [0,1]^n$ and for each index $i\leq n$, $\agg^{(n)}_{\overline{\r},{-i}}$ is a linear function. That is, \begin{align*}
        \agg^{(n)}_{\overline{\r},{-i}} (x) &= \agg^{(n)} (\r_1, \ldots, \r_{i-1}, x, \r_{i+1},\ldots, \r_n)\\
        &= ax + b \qquad \text{ for some }a,b\in \mathbb R. \end{align*}  
\end{description}

The third set of axioms concerns the behaviour of the aggregation function under certain special input values and can be interpreted as follows. Firstly, we can require that having been assigned full member contribution can not be relativized by other actions. 
\kick{In this case 1 is called an `annihilator' of the aggregation function.}
On the other hand, we can demand that adding a decision which had a member contribution of 0 does not change an agents', or their groups', assigned outcome responsibility. \kick{In this case, we say that 0 is a `neutral element' of the aggregation function.}
Lastly, one could require that doing something ``equally bad'' multiple times does not change the ascribed outcome responsibility.\kick{This would be called `idempotency' of the aggregation function
and stands in contrast to (BSM$^+$) defined above, which demands that if we add non-zero entries, the aggregated responsibility increases. In general, non-compliance with idempotency might be more desirable than compliance, depending on the modelled scenario.
Idempotency is for example relevant in situations as depicted in \Cref{fig:repeated-shooting} that we will discuss later in this section. }
\begin{description}
    \item[(AN1)] \emph{Annihilator 1.} For each $n\in \mathbb{N}_{\geq 1}$ and all $[x_1, \ldots, x_n] \in [0,1]^n$ it holds that if $x_k = 1$ for some $k \in \{1, \ldots, n\}$ then 
    \[\agg([x_1, \ldots, x_n]) = 1.\]
    
    \item[(NE0)] \emph{Neutral element 0.} For each $n \in \mathbb{N}_{\geq 1}$, all $[x_1, \ldots, x_n] \in [0,1]^n$ such that $x_i = 0$ for some $i \in \{1, \ldots, n\}$ it holds that 
    \[\agg([x_1, \ldots, x_n]) = \agg([x_1, \ldots, x_{i-1}, x_{i+1}, \ldots, x_n]).\]
    
    \item[(SIP)] \emph{Strong idempotency.} For all $k \in \mathbb{N}$ and $\x \in \bigcup_{n \in \mathbb{N}_{\geq 1}} [0,1]^n$ it holds that 
    \[\agg(\underbrace{\x \oplus \ldots \oplus \x}_\text{$k$ times}) = \agg(\x),\] 
    where $\x \oplus \ldots \oplus \x \in [0,1]^{kn}$ is the vector obtained by concatenating $\x$ with itself $k$ times. 
\end{description}

The next axiom, \emph{anonymity}, requires that it does not matter who exactly makes a certain decision, only the `value' of their decision 
matters for the final evaluation. Additionally, one can require that the time of the decision making has no effect. 

\begin{description}
    \item[(AAT)] {\em Anonymity (with respect to agents and time).} For any $\mathbf{x} \in [0,1]^n$ 
    and any permutation $\pi$ of the entries in $\mathbf x$ it holds that
        $ \agg(\mathbf x) = \agg( \pi (\mathbf x)). $
\end{description}

Finally, we consider an axiom which seems very natural if we hold that outcome responsibility and member contribution are intended to capture aspects of the same concept: that outcome responsibility is to be zero exactly in those cases where all member contributions which we provide as an input are also zero.

\begin{description}
    \item[(RED)] {\em Reducibility.} For any $\mathbf x \in [0,1]^n$ it holds that 
    $\agg(\mathbf x) = 0 \text{ iff }\mathbf x = (0,\ldots, 0).$
\end{description}

\subsection{Proposed Aggregation Functions}\label{sec:agg-functions}

Having developed a set of possibly desirable properties for aggregation functions we now continue by suggesting specific functions. We will discuss the full matrix of axiom compliance in \Cref{sec:agg-evaluation}. 
Let a tree $\T$, group of agents $G\subseteq I$ and outcome node $w \in V_o$ be given. We drop $\T$ and $G$ from the notation.
Note that the first proposed aggregator, the sum, is technically not an aggregation function in the sense of our definition in \Cref{sec:framework}, as it might exceed the interval $[0,1]$. We still include it here as it is a natural way to aggregate member contributions.

\paragraph{Variant 1, sum.} 
 \[ \mathbf{sum} [(\r(i, v, \mathtt{a}_{v \to w} ))_{i\in G, v \in V_i\cap H(w)} ] := \sum_{i\in G, v\in V_i \cap H(w)}  \r(i, v, \mathtt{a}_{v \to w}) \]

\paragraph{Variant 2, average.}
\[ \mathbf{avg} [(\r(i, v, \a_{v \to w} ))_{i\in G, v \in V_i\cap H(w)} ] := \frac{1}{|\Set{v\in V_i \cap H(w) \mid i\in G}|} \sum_{i\in G, v\in V_i \cap H(w)}  \r(i, v, \a_{v \to w}) \]
 
\paragraph{Variant 3, maximum.}
  \[ \mathbf{max} [(\r(i, v, \mathtt{a}_{v \to w} ))_{i\in G, v \in V_i\cap H(w)} ] := \max_{i\in G, v\in  V_i\cap H(w)}  \r(i,v, \mathtt{a}_{v \to w}) \]

 \paragraph{Variant 4, modified product.}
  \[ \mathbf{mprod} [(\r(i, v, \mathtt{a}_{v \to w} ))_{i\in G, v \in V_i\cap H(w)} ] := 1- \prod_{i\in G, v\in V_i \cap H(w)} (1-  \r(i, v, \mathtt{a}_{v \to w})). \]

\subsection{Evaluation of Proposed Aggregation Functions}\label{sec:agg-evaluation}

In addition to the axioms introduced in \Cref{sec:agg},  we include one more property in our axiomatic analysis. As we noted above, the sum technically is not an aggregation function. To make this explicit we add the following property that is satisfied by all aggregation functions in the sense of the definition in \Cref{sec:framework} but is violated by the sum.
\begin{description}
    \item[(01B)]{\em 0,1-boundedness.} $\agg (\bigcup\limits_{n\in \mathbb N_{\geq 1}}[0,1]) \subseteq [0,1]$.
\end{description}

\begin{proposition}\label{prop:agg_compliance}
    Compliance of the candidate aggregation functions with the axioms from \Cref{sec:agg-axioms} and (01B) is as depicted in \Cref{tbl:axioms-agg}.
\end{proposition}
\noindent
We abstain from discussing the proof here and instead refer to \Cref{sec:appendix-evaluation-agg}.

\begin{table}
\centering
\footnotesize
\begin{tabular}{lccccccccc} \toprule
\phantom{\textbf{Variant}}    & (01B)  & (BSM$^+$) & (BSM$^>$) & (LIN) & (AN1) & (NE0) & (SIP) &  (AAT) & (RED)   \\ 
\midrule
                     $\mathbf{sum}$     & \checkmark    & \checkmark    & \checkmark   & \checkmark & $\times$   & \checkmark    & $\times$     & \checkmark &  \checkmark \\
\rowcolor{lightgray} $\mathbf{avg}$     & \checkmark    & $\times$      & \checkmark   & \checkmark & $\times$   & $\times$      & \checkmark   & \checkmark & \checkmark \\ 
                     $\mathbf{max}$     & \checkmark    & $\times$      & $\times$     & $\times$   & \checkmark & $\times$      & \checkmark   & \checkmark & \checkmark \\
\rowcolor{lightgray} $\mathbf{mprod}$   & \checkmark    & \checkmark    & \checkmark   & \checkmark & \checkmark & \checkmark    & $\times$     & \checkmark & \checkmark \\
\bottomrule
\end{tabular}
\caption{\label{tbl:axioms-agg}
Summary of axiom compliance by the proposed aggregation functions.}
\end{table}

The sum fulfills monotonicity, anonymity and reducibility. However, it fails 0,1-boundedness and thus is not an aggregation function. Averaging the sum remains 0,1-bounded but fails monotonicity under adding entries. As the same reasoning applies to other normalisation methods we move away from considering normalisations of the sum. We mention the maximum for reasons of completeness, despite it clearly not fulfilling any monotonicity requirements. The same would hold for the product, which is why it is not not mentioned in detail and does not appear in the table. However, a modification of the product seems promising:
The issue with the first three proposed functions seems to be that they do not correctly include additional contributions when considering several actions, thus either failing 0,1-boundedness or monotonicity. The rationale behind the modified product becomes clearer when we look at an example of a repeated action. Consider the scenario depicted in \Cref{fig:repeated-shooting}, and consider the path where agent $i$ selects continue twice and then repair. If we assume that continuing once with a 90\% chance of disaster produces a member contribution of $0.9$, then in the given example continuing the first time adds $0.9$ to the outcome responsibility, continuing the second times adds $0.9$ of the remaining interval (i.e. $0.09$), and the final decision to repair does not reduce any of the previously gained contributions. 
This example also shows that the modified product does not fulfill (SIP) and hints at non-compliance with (SIP) being in fact desirable in certain cases.

From the proposed aggregation functions, $\mathbf{mprod}$ is the only one that is a proper aggregation function for any input member contribution values and which fulfills all the described axioms except for (SIP).

Moreover, we can even show the stronger statement that $\mathbf{mprod}$ is characterised by a combination of the discussed axioms and properties.\footnote{We require (LIN) for this characterisation. As one might argue against this axiom due to it being very strong and rather technical it is important to point out that even without the characterisation result, $\mathbf{mprod}$ is the most highly axiom compliant function out of the ones we suggested.}

\begin{restatable}{theorem}{characterisationmprod}\label{thm:characterisation-mprod}
The axioms (LIN), (AAT), (NE0) and (AN1) together uniquely characterise $\mathbf{mprod}$ among all aggregation functions.
\end{restatable}
\noindent
We will only present the main idea of the proof here. For the full details we refer to \Cref{sec:appendix-evaluation-agg}.

\begin{proof}[Sketch]
Linearity in each component together with a subset of the other axioms constrains the function to 
\begin{align*}
\agg([ &\r_1,\ldots,\r_n]) = \\
&b+ \sum\limits_{i=1}^n a_i \r_i + \sum\limits_{\substack{i,i'= 1 \\ i<i'}}^n a_{i,i'} \r_i \r_{i'} + \ldots + \sum\limits_{\substack{i^{(1)},\ldots, i^{(n-1)} = 1  \\ i^{(1)}<\ldots < i^{(n-1)}}}^n a_{i^{(1)},\ldots,i^{(n)}} \r_{i^{(1)}}\ldots\r_{i^{(n-1)}}
+ a_{1,\ldots,n} \r_1\ldots\r_n. 
\end{align*}
Further axioms restrict the values of the coefficients as follows.
% (RED) implies that $b=0$, 
(AAT) implies that $a_1 = \ldots = a_n =: a^{(1)}$, $a_{1,2}= \ldots = a_{n-1, n} =:a^{(2)}$ etc.
Since an aggregation function is the identity on $[0,1]$, we know that $a^{(1)} = 1$, (NE0) additionally implies that $b=0$ and with $(AN1)$ we can inductively prove that $a^{(k)} = (-1)^{k-1} \quad \forall k\geq 1$, which results precisely in the function $\mathbf{mprod}$.
\end{proof}

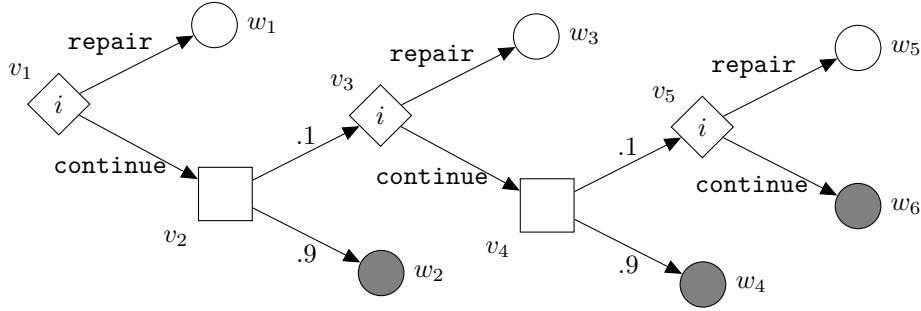
\begin{figure}
\centering
\begin{tikzpicture} [->, scale=1, every node/.style={scale=1}, node distance=1cm, decision node/.style={diamond, draw, aspect=1, minimum height= 7mm},  prob node/.style={regular polygon,regular polygon sides=4, draw, minimum height= 1cm},  dummy/.style = {}, good outcome/.style={circle, draw, minimum width=6mm}, bad outcome/.style={circle, draw, fill=gray, minimum width=6mm }, >= triangle 45]
\node[decision node] (v1) [] [label=above left: $v_1$] {$i$};

\node[dummy] (a) [right of = 1] [] {};
%\node[dummy] (a2) [right of = a] [] {};

\node[good outcome] (w1) [above right =  of a] [label=right: $w_1$] { };
\node[prob node] (p1) [below right = of a] [label=below left: $v_2$] { };

\node[dummy] (b) [right of = p1] [] {};
%\node[dummy] (b2) [right of = b] [] {};
%\node[dummy] (b3) [right of = b2] [] {};

\node[decision node] (v2) [above right =  of b] [label=above left: $v_3$]{$i$};
\node[bad outcome] (w2) [below right =  of b] [label=right: $w_2$] {};

\node[dummy] (c) [right of = v2] [] {};
%\node[dummy] (c2) [right of = c] [] {};

\node[good outcome] (w3) [above right = of c] [label=right: $w_3$] {};
\node[prob node] (p2) [below right = of c] [label=below left: $v_4$] {};

\node[dummy] (d) [right of = p2] [] {};
%\node[dummy] (c2) [right of = c] [] {};

\node[decision node] (v3) [above right = of d] [label=above left: $v_5$] {$i$};
\node[bad outcome] (w4) [below right = of d] [label=right: $w_4$] {};

\node[dummy] (e) [right of = v3] [] {};
%\node[dummy] (c2) [right of = c] [] {};

\node[good outcome] (w5) [above right = of e] [label=right: $w_5$] {};
\node[bad outcome] (w6) [below right = of e] [label=right: $w_6$] {};

\path (v1) edge node[above] {$\mathtt {repair \qquad}$} (w1) (v1) edge node[below] {$\mathtt{continue \, \qquad}$} (p1) (p1) edge node[above] {$.1$}  (v2) (p1) edge node[below] {$.9$} (w2) (v2) edge node[above] {$\mathtt {repair \qquad}$} (w3) (v2) edge node[below] {$\mathtt{continue \, \qquad}$} (p2) (p2) edge node[above] {$.1$} (v3) (p2) edge node[below] {$.9$} (w4) (v3) edge node[above] {$\mathtt {repair \qquad}$} (w5) (v3) edge node[below] {$\mathtt{continue \, \qquad}$} (w6);

\end{tikzpicture}
\caption{\label{fig:repeated-shooting} Representation of an example scenario of a repeated action. The operator of a large machine can either ignore a warning of the machine and continue using it, knowing that this has a high chance of irreversible damage or perform a repair. At some point in time the last chance to repair the machine is reached.}
\end{figure}

\section{Outcome Responsibility Functions}\label{sec:outcomeresp}

After exploring member contribution functions and aggregation functions on their own we will now turn our attention to outcome responsibility as a combination of the two previously mentioned classes of functions. Again, we first propose a set of axioms that capture interesting properties before checking axiom compliance for a set of proposed functions.

\subsection{Axioms for Outcome Responsibility Functions}\label{sec:outcomeresp-axioms}

The following four axioms give upper and lower bounds on the assigned responsibility and in that sense follow a similar rationale as (AMC) and (FMC) for member contribution functions. 
The first two axioms present upper bounds by determining that there are cases in which no responsibility is to be ascribed. The first one captures the idea that if an agent has complete control over the situation and steers it towards a desired outcome, that agent should not be assigned responsibility.
The second axiom requires that agents have initial strategies which avoid responsibility.
The fact that there is a qualitatively different action available
is seen by some authors as one of the fundamental criteria for assigning responsibility. For example Braham and van Hees \cite{braham2018} introduce an \emph{avoidance opportunity condition}, which requires that a person has an opportunity to avoid contributing to an outcome in order for responsibility ascription to be appropriate. A similar -- but not equivalent -- idea is captured by the second axiom. 

\begin{description}
\item[(CC)]{\em Complete control.}
  Given a tree $\T$, a singleton group $G = \{i\}$ and a desirable outcome node $w\in V_o \setminus \epsilon$ such that the path to $w$ involves only decisions of agent $i$ and no decision nodes of other agents nor ambiguity or probability nodes, $H(w) \subseteq V_i$. Then $\R (w, i) = 0$.
  
\item[(NUR)] {\em No unavoidable responsibility.}
  Given a tree $\T$, each group $G\subseteq I$ must have an original strategy 
  that is guaranteed to avoid any outcome responsibility. That is, denoting the root node of $\T$ by $v_r$, for every group of agents $G\subseteq I$ there must be a strategy $\sigma\in\Sigma(\T,G,v_r)$, such that $\R(G,w) = 0$ for all outcomes $w\in V_o\cap \sigma(V_G\cap B^{\sim}(v_r))$ reachable with that strategy.
\end{description}

The next two axioms are lower bounds in the sense that they require a certain amount of responsibility to be attributed. Specifically, they request the absence of (individual or group-wise) responsibility voids.

\begin{description}
\item[(NRV)] {\em No responsibility voids.}
  Given a tree $\T$ with no uncertainty nodes, $V_a=V_p=\emptyset$,
  and a set of undesirable outcomes $\epsilon\neq V_o$.
  Then for each undesirable outcome $w\in\epsilon$, 
  some group is at least partially responsible, i.e., there exists a $G\subseteq I$ with $\R(G,w) > 0$.

\item[(NIRV)]{\em No individual responsibility voids.}
  Given a tree $\T$ with no uncertainty nodes, $V_a=V_p=\emptyset$,
  and a set of undesirable outcomes $\epsilon\neq V_o$. Then for each undesirable outcome $w\in\epsilon$, there is an individual agent $i \in I$ with $\R(\{i\},w)>0$.
  \kick{This axiom is a strengthening of NRV to only consider single agents, i.e., it rules out that each individual voter is assigned zero responsibility.}
\end{description}
\noindent
Note that as we allow group composition to inform member contribution functions it is not necessarily the case that (NRV) together with (RED) implies (NIRV).

The absence of individual responsibility in the presence of group responsibility is discussed repeatedly in the literature. 
That is, cases where (NRV) holds, but (NIRV) for members of the same group does not. 
For example, in the presentation by Duijf \cite{duijf2018} the premises we provide in our axioms place us in a cooperative decision context with no external uncertainty but with coordination uncertainty. This combination is shown to allow for the absence of individual responsibility in the presence of group responsibility. Similarly, Braham and van Hees \cite{braham2018} show that this combination is possible if a mechanism does not allow for the attribution of different aspects of an outcome to varying individuals but only for an overall attribution of responsibility at large, which is the case in our representation. 

In addition to the above axioms that are specifically designed for outcome responsibility functions, we also say that an outcome responsibility function satisfies (KSym), (AMC) or (FMC) -- as well as their weaker versions (AMC$^\sim$) and (FMC$^\sim$) -- if the underlying member contribution function satisfies the respective axiom.
Using this we prove the following impossibilities.
All proofs omitted in this section can be found in \Cref{app:outcome_proofs}.

\begin{restatable}{proposition}{NURvsNRV}\label{thm:NURvsNRV}
For any outcome responsibility function $\R$ the following impossibilities hold.
\begin{itemize}
    \item[(i)] $\R$ can not simultaneously satisfy (KSym), (FMC), and (NUR). 
    \item[(ii)] $\R$ can not simultaneously satisfy (KSym), (NRV), and (NUR).
\end{itemize}
\end{restatable}

\subsection{Evaluation}\label{sec:outcomeresp-evaluation}

In the following we show the compliance of certain classes of responsibility functions with these new axioms, with the full overview being visible in \Cref{tbl:axioms_resp_2}. 

\begin{table}
\centering
\small
\begin{tabular}{lcccc|ccc} \toprule
                        \phantom{\textbf{Variant}}    & (NRV)        & (NIRV)        & (NUR)         & (CC)       & (KSym)    & (AMC)     & (FMC)       \\ \midrule
                       % $\rstrict$              & \checkmark   & \checkmark    & \checkmark    & \checkmark  \\
                        $\agg \circ \rlike$               & \checkmark   & \checkmark    & \checkmark    & \checkmark   & $\times$      & \checkmark    & \checkmark\\
\rowcolor{lightgray}    $\agg \circ \rrisk$                & \checkmark   & \checkmark    & $\times$      & \checkmark  & \checkmark    & $\times$      & \checkmark\\
                        $\agg \circ \rneg$                 & $\times$     & $\times$      & \checkmark    & \checkmark  & \checkmark    & \checkmark    & $\times$ \\ \bottomrule
\end{tabular}
\caption{\label{tbl:axioms_resp_2}
Summary of axiom compliance by the three variants of member contribution from \Cref{sec:pwresp-axioms} combined with an aggregation functions $\agg$ that satisfies (BSM) and (RED). For the axioms initially defined for member contribution functions on the right hand side of the table we say that the combined function satisfies the axiom if its member contribution function satisfies it.}
\end{table}

\begin{restatable}{proposition}{evalcombs}\label{prop:eval_combs}
Let $\agg$ be an aggregation function fulfilling (BSM) and (RED) and let $\r^{\triangle}$ for $\triangle \in \Set{\mathrm{like}, \mathrm{risk}, \mathrm{negl}}$ be as defined earlier. Then the compliance of the responsibility function $\R^\triangle := \agg \circ \r^\triangle$ with the axioms {(NRV)}, {(NIRV)}, {(NUR)} and {(CC)} is as presented in \Cref{tbl:axioms_resp_2}.
\end{restatable}

\noindent
This result immediately leads us to the following conclusion.

\begin{corollary}
Given a tree $\T$, a group of agents $G \subseteq I$ and an outcome node $w\in V_o$. Then the following two statements hold.
    \begin{itemize}
        \item[(i)] $\mathbf{mprod} \circ \rrisk $ fulfills all the axioms we discussed in \Cref{sec:agg} (except for (SIP)) as well as (KSym), (FMC), (NRV), (NIRV), and (CC).
        \item [(ii)] $\mathbf{mprod} \circ \rneg $ fulfills all the axioms we discussed in \Cref{sec:agg} (except for (SIP)) as well as (KSym), (AMC), (NUR), and (CC).
    \end{itemize}
\end{corollary}

The results presented here show that if we constrain ourselves to functions that satisfy (KSym), we have to choose between fulfilling (NUR) and thus ensure some upper bound on the ascribed responsibility or fulfilling (NIRV) and (FMC) and ensure some lower bound. The two combinations discussed in the above lemma provide solutions to both choices, while being maximally compliant with the other presented axioms.

\subsection{Application to Example Scenario}

In the example presented in \Cref{fig:repeated-shooting} we obtain the following outcome responsibility values using the two functions determined above: $\mathbf{mprod}\circ \rrisk$ and $\mathbf{mprod} \circ \rneg$. The results are summarised in \Cref{tbl:outcome-application}. Note that as selecting $\mathtt{repair}$ always results in a member contribution of zero, $\rrisk$ and $\rneg$ actually return the same results.

\begin{table}
\centering
\small
\begin{tabular}{l C{1.2cm}C{1.2cm}C{1.2cm}C{1.2cm}C{1.2cm}C{1.2cm}} \toprule
  & $w_1$      & $w_2$       & $w_3$ & $w_4$ & $w_5$    & $w_6$   \\ \midrule
$\mathbf{mprod} \circ \rlike$ & 0  & 0.9   & 0.9  & 0.99  & 0.99     & 1    \\
\rowcolor{lightgray}  $\mathbf{mprod} \circ \rrisk$                & 0   & 0.9    & 0.9     & 0.99  & 0.99  & 1\\
\bottomrule
\end{tabular}
\caption{\label{tbl:outcome-application}
Summary of the outcome responsibility values for all outcomes in the example in \Cref{fig:repeated-shooting}.}
\end{table}

\section{Conclusion}\label{sec:summary}

We employed an axiomatic analysis of responsibility quantification in complex decision situations containing interactions and uncertainty.
The model is based on representing appropriate decision scenarios via an extension of extensive form game trees and assigning (moral) responsibility for a certain outcome to a group of agents as a value between 0 and 1, based on the change in the probability of an undesirable outcome.

We split responsibility functions for outcomes into member contribution functions and aggregation functions. The former can be viewed as functions for member contribution assigned to individuals as part of a group, while the latter determine the aggregation of these individual values into a group evaluation.
This separation enables a more fine-grained analysis of the functions' properties. 
We translated different kinds of desiderata from the literature on (moral) responsibility into this setting and formulated respective axioms for the contributing functions. 

We presented an impossibility result for three axioms regarding member contribution. This elucidates the intuition that a principled and sensible assignment of responsibility in coordination games is impossible. 
We were then able to single out two promising candidates that satisfy the two opposing points of view in situations such as the coordination game, and that satisfy all of the other considered axioms for member contribution functions.

We were able to derive a characterisation of an appealing aggregation function, the modified product, using five intuitive axioms. 
For outcome responsibility functions we prove compliance with respective axioms for classes of combinations of member contribution functions and aggregation functions. The previously mentioned impossibility result transfers to the case of outcome responsibility. When we require that agents are assigned the same member contribution in all situations they cannot discern between,
then a responsibility function has to either allow for voids or for unavoidable responsibility. The same pair of member contribution functions which we determined earlier, paired with the selected aggregation function, is able to capture each side of this division.

The work presented in this paper can be continued in various directions.
First of all, we can use the functions which we determined here as fulfilling most of the desired properties in real-world application scenarios. Next, a characterisation of the space of all three functions using appropriate axioms would be ideal.
Another promising direction for future research is the introduction of weighted aggregation functions that are not indifferent to the time at which a decision was taken. This enables the study of responsibility in settings where more recent actions or decisions should carry more weight than more distant ones. 

\subsubsection*{Acknowledgements.}
We thank % alphabetical:
Alexandru Baltag,
Markus Brill, 
Rupert Klein, 
the rest of the Efficient Algorithm research group at Technical University Berlin, 
and 
the copan collaboration at the Potsdam Institute for Climate Impact Research
for fruitful discussions and %helpful 
comments.

The research leading to these results received funding from the Deutsche Forschungsgemeinschaft (DFG, German Research Foundation) under Germany's Excellence Strategy -- The Berlin Mathematics Research Center MATH+ (EXC-2046/1, project ID: 390685689) and under grant BR~4744/2-1.

\subsubsection*{Statements and Declarations.}
The authors have no financial or proprietary interests in any material discussed in this article.

\bibliographystyle{splncs04}
\bibliography{references.bib}

\clearpage

\appendix

\section{Example Evaluations of Member Contribution Functions}\label{app:pw_examples}

Recall the examples.
\renewcommand{\thefigure}{1}
\begin{figure}
\centering
{
    {\bf (a)}\begin{tikzpicture} [->, scale=.9, every node/.style={scale=.9}, node distance=1cm, decision node/.style={diamond, draw, aspect=1, minimum height= 7mm},  prob node/.style={regular polygon,regular polygon sides=4, draw, minimum height= 1cm},  dummy/.style = {}, good outcome/.style={circle, draw, minimum width=6mm}, bad outcome/.style={circle, draw, fill=gray, minimum width=6mm }, >= triangle 45]
\node[decision node] (1) [] [label=above left: $v_1$] {$i$};

\node[dummy] (a) [right of = 1] [] {};
\node[dummy] (a2) [right of = a] [] {};

\node[prob node] (2) [above right =  of a2] [label=above left: $v_2$] { };
\node[decision node] (3) [below right = of a2] [label=below left: $v_3$] { \ \ \ \ \ };

\node[dummy] (b) [right of = 2] [] {};
\node[dummy] (b2) [right of = b] [] {};
\node[dummy] (b3) [right of = b2] [] {};

\node[good outcome] (good2) [above right = .4cm of b2] [label = right: $w_1$]{}; %label = right: $v_1$
\node[bad outcome] (bad2) [below right = .4cm of b2] [label = right: $w_2$] {};

\node[dummy] (c) [right of = 3] [] {};
\node[dummy] (c2) [right of = c] [] {};

\node[good outcome] (good3) [above right = .4cm of c2] [label = right: $w_3$] {};
\node[bad outcome] (bad3) [below right = .4cm of c2] [label = right: $w_4$] {};

\path (1) edge node[above] {$\mathtt{evade \ \ }$} (2) (1) edge node[below] {$\mathtt{steady \ \ }$} (3) (2) edge node[above] {$p$}  (good2) (2) edge node[below] {$1 - p$} (bad2) (3) edge node[above=5pt] {$\mathtt{deer \ moves\ \ }$} (good3) (3) edge node[below=3pt] {$\mathtt{deer \ stays\ \ \ }$} (bad3) ;

\end{tikzpicture}
}
{
    {\bf (b)}
    \begin{tikzpicture} [-, scale=.9, every node/.style={scale=.9}, node distance=1cm, decision node/.style={diamond, draw, aspect=1, minimum height= 7mm},  prob node/.style={regular polygon,regular polygon sides=4, draw, minimum height= 1cm},  dummy/.style = {}, good outcome/.style={circle, draw, minimum width=6mm}, bad outcome/.style={circle, draw, fill=gray, minimum width=6mm }, >= triangle 45]
\node[decision node] (1) [] [label=above left: $v_1$] {$i$};

\node[dummy] (a) [right of = 1] [] {};
\node[dummy] (a2) [right of = a] [] {};

\node[decision node] (2) [above right =  of a2] [label=above left: $v_2$] {$j$}; %
\node[decision node] (3) [below right = of a2] [label=below left: $v_3$] {$j$};

\node[dummy] (b) [right of = 2] [] {};
\node[dummy] (b2) [right of = b] [] {};
\node[dummy] (b3) [right of = b2] [] {};

\node[bad outcome] (good2) [above right = .4cm of b2] [label = right: $w_1$]{};
\node[good outcome] (bad2) [below right = .4cm of b2] [label = right: $w_2$] {};

\node[dummy] (c) [right of = 3] [] {};
\node[dummy] (c2) [right of = c] [] {};

\node[good outcome] (bad3) [above right = .4cm of c2] [label = right: $w_3$] {};
\node[good outcome] (good3) [below right = .4cm of c2] [label = right: $w_4$] {};

\path[->] (1) edge node[above] {$\mathtt{ignore}$}  (2) (1) edge node[below] {$\mathtt{test}$} (3) (2) edge node[above] {$\mathtt{ignore}$}  (good2) (2) edge node[below] {$\mathtt{test}$} (bad2) (3) edge node[below] {$\mathtt{test}$} (good3) (3) edge node[above] {$\mathtt{ignore}$} (bad3);

\draw[line width=2pt, dotted=on] (2) -- (3); 

\end{tikzpicture}}

 \caption{(a) An autonomous vehicle (agent $i$) approaches a deer on a street. It can either try to evade the deer with a computed success rate of $p$ or keep steady, resulting in a crash if the deer does not move.
 (b) $i$ and $j$ are two instances in a distributed system which may perform a costly test for the failure of a component, with $i$ coming before $j$. Both may either perform the test or pass, but $j$ has no way of knowing whether $i$ performed the test. If they both pass a failure is overlooked.
 }
\end{figure}
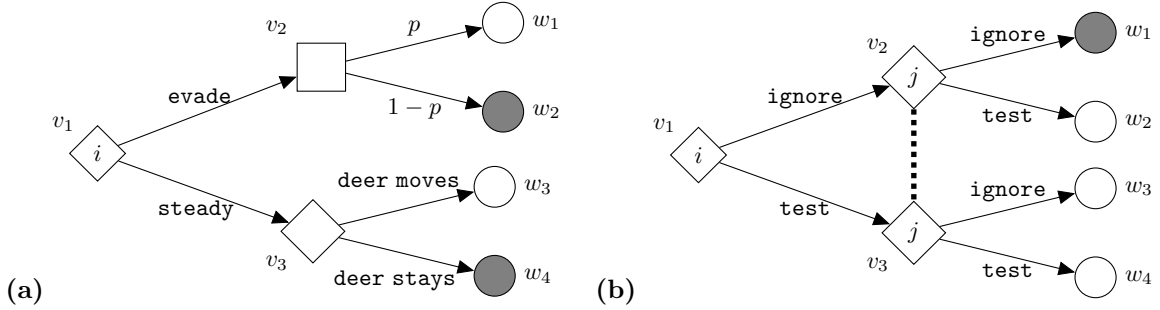

\renewcommand{\thefigure}{2}
\begin{figure}
\centering
% \subfigure{
    \begin{tikzpicture} [-, scale=.9, every node/.style={scale=.9}, node distance=1cm, decision node/.style={diamond, draw, aspect=1, minimum height= 7mm},  prob node/.style={regular polygon,regular polygon sides=4, draw, minimum height= 1cm},  dummy/.style = {}, good outcome/.style={circle, draw, minimum width=6mm}, bad outcome/.style={circle, draw, fill=gray, minimum width=6mm }, >= triangle 45]
\node[decision node] (1) [] [label=above left: $v_1$] {$i$};

\node[dummy] (a) [right of = 1] [] {};
\node[dummy] (a1) [right of = a] [] {};
\node[dummy] (a2) [right of = a1] [] {};

\node[decision node] (2) [above right =  of a2] [label=above left: $v_2$] {$j$};
\node[decision node] (3) [below right = of a2] [label=below left: $v_3$] {$j$};

\node[dummy] (b) [right of = 2] [] {};
\node[dummy] (b1) [right of = b] [] {};
\node[dummy] (b2) [right of = b1] [] {};
\node[dummy] (b3) [right of = b2] [] {};

\node[good outcome] (good2) [above right = .4cm of b2] [label = right: $w_1$]{};
\node[bad outcome] (bad2) [below right = .4cm of b2] [label = right: $w_2$] {};

\node[dummy] (c) [right of = 3] [] {};
\node[dummy] (c1) [right of = c] [] {};
\node[dummy] (c2) [right of = c1] [] {};

\node[bad outcome] (bad3) [above right = .4cm of c2] [label = right: $w_3$] {};
\node[good outcome] (good3) [below right = .4cm of c2] [label = right: $w_4$] {};

\path[->] (1) edge node[above] {$\mathtt{left}$}  (2) (1) edge node[below] {$\mathtt{right}$} (3) (2) edge node[above] {$\mathtt{left}$}  (good2) (2) edge node[below] {$\mathtt{right}$} (bad2) (3) edge node[below] {$\mathtt{right}$} (good3) (3) edge node[above] {$\mathtt{left}$} (bad3);

\draw[line width=2pt, dotted=on] (2) -- (3); 

\end{tikzpicture}

\caption{A coordination game where two autonomous vehicles (agents $i$ and $j$) approach each other on a street with two lanes, destined to crash into each other if they don't evade. The crash is avoided if and only if both vehicles move to their respective right or both move to their respective left side.}
\end{figure}

We will now compute the member contribution values for these examples given by the three different functions proposed in \Cref{sec:pwresp-functions}. 

\paragraph{$\T$ as in \Cref{fig:scenarios} (a):}

\begin{align*}
    \rlike (\{ i\} , i , v_1, \mathtt{evade})&= \Delta \gamma (v_1, \mathtt{evade}) = \gamma (v_2) - \gamma (v_1) \\
    &= \min_{\sigma \in \Sigma(\T, \{i\}, v_2)} \min_{\zeta \in Z^\sim(\T,\{i\},v_2)} \ell(\epsilon | v_2,\sigma, \zeta) - \min_{\sigma \in \Sigma(\T, \{i\}, v_1)} \min_{\zeta \in Z^\sim(\T,\{i\},v_1)} \ell(\epsilon | v_1,\sigma, \zeta) \\
    &= (1-p)-0 = 1-p
\end{align*}

\begin{align*}
    \rlike (\{ i\} , i , v_1, \mathtt{steady})&= \Delta \gamma (v_1, \mathtt{evade}) = \gamma (v_3) - \gamma (v_1) \\
    &= \min_{\sigma \in \Sigma(\T, \{i\}, v_3)} \min_{\zeta \in Z^\sim(\T,\{i\},v_3)} \ell(\epsilon | v_3,\sigma, \zeta) - \min_{\sigma \in \Sigma(\T, \{i\}, v_1)} \min_{\zeta \in Z^\sim(\T,\{i\},v_1)} \ell(\epsilon | v_1,\sigma, \zeta) \\
    &= 0-0 = 0
\end{align*}

\begin{align*}
    \rrisk (\{ i\} , i , v_1, \mathtt{evade})&= \max_{\zeta \in Z^\sim(\T, \{i\}, v_1)} \Delta \omega(v_1, \zeta, \mathtt{evade}) = \max_{\zeta \in Z^\sim(\T, \{i\}, v_1)} [\omega (v_2, \zeta) - \omega(v_1, \zeta)] \\
    &= \max_{\zeta \in Z^\sim(\T, \{i\}, v_1)} [\min_{\sigma\in\Sigma(\T, \{i\}, v_2)} \ell(\epsilon | v_2,\sigma,\zeta) - \min_{\sigma\in\Sigma(\T, \{i\}, v_1)} \ell(\epsilon | v_1,\sigma,\zeta)]\\
    &= \max_{\zeta \in \{\mathtt{deer\ moves, deer\ stays}\}} [\min_{\sigma\in\Sigma(\T, \{i\}, v_2)} \ell(\epsilon | v_2,\sigma,\zeta) - \min_{\sigma\in\Sigma(\T, \{i\}, v_1)} \ell(\epsilon | v_1,\sigma,\zeta)]\\
    &= \max\{ [(1-p) - \min_{\sigma\in\Sigma(\T, \{i\}, v_1)} \ell(\epsilon | v_1,\sigma,\mathtt{deer\ moves})],  \\
    & \ [(1-p) - \min_{\sigma\in\Sigma(\T, \{i\}, v_1)} \ell(\epsilon | v_1,\sigma,\mathtt{deer\ stays})]\} \\
    &= \max \{ (1-p) - 0, (1-p) - (1-p) \} = 1-p 
\end{align*}

\begin{align*}
    \rrisk (\{ i\} , i , v_1, \mathtt{steady})&= \max_{\zeta \in Z^\sim(\T, \{i\}, v_1)} \Delta \omega(v_1, \zeta, \mathtt{steady}) = \max_{\zeta \in Z^\sim(\T, \{i\}, v_1)} [\omega (v_3, \zeta) - \omega(v_1, \zeta)] \\
    &= \max_{\zeta \in Z^\sim(\T, \{i\}, v_1)} [\min_{\sigma\in\Sigma(\T, \{i\}, v_3)} \ell(\epsilon | v_3,\sigma,\zeta) - \min_{\sigma\in\Sigma(\T, \{i\}, v_1)} \ell(\epsilon | v_1,\sigma,\zeta)]\\
    &= \max_{\zeta \in \{\mathtt{deer\ moves, deer\ stays}\}} [\min_{\sigma\in\Sigma(\T, \{i\}, v_3)} \ell(\epsilon | v_3,\sigma,\zeta) - \min_{\sigma\in\Sigma(\T, \{i\}, v_1)} \ell(\epsilon | v_1,\sigma,\zeta)]\\
    &= \max\{ [0 - \min_{\sigma\in\Sigma(\T, \{i\}, v_1)} \ell(\epsilon | v_1,\sigma,\mathtt{deer\ moves})],  \\
    & \ [1 - \min_{\sigma\in\Sigma(\T, \{i\}, v_1)} \ell(\epsilon | v_1,\sigma,\mathtt{deer\ stays})]\} \\
    &= \max \{ 0 - 0, 1 - (1-p) \} = p 
\end{align*}

Assume $p < 1-p$, thus $\underline\rho(\{i\} , i, v_1) = \min\limits_{\a \in A_{v_1}} \rrisk(\{i\}, i, v_1, \a) = p $.

\begin{align*}
    \rneg(\{i\}, i, v_1, \mathtt{evade}) &=  \max_{\zeta \in Z^\sim(\T, \{i\}, v_1)} \Delta \omega(v_1, \zeta, \mathtt{evade}) - \underline\rho(v_1) \\
    &= \rrisk (\{ i\} , i , v_1, \mathtt{evade}) - \underline\rho(v_1) = (1-p) - p = 1-2p
\end{align*}

This is larger than $0$, as we assumed that $p< 1-p$.

\begin{align*}
    \rneg(\{i\}, i, v_1, \mathtt{steady}) &=  \max_{\zeta \in Z^\sim(\T, \{i\}, v_1)} \Delta \omega(v_1, \zeta, \mathtt{steady}) - \underline\rho(v_1) \\
    &= \rrisk (\{ i\} , i , v_1, \mathtt{steady}) - \underline\rho(v_1) = p - p = 0
\end{align*}

\paragraph{$\T$ as in \Cref{fig:scenarios} (b):}

\begin{align*}
    \rlike (\{ i\} , i , v_1, \mathtt{ignore})&= \Delta \gamma (v_1, \mathtt{ignore}) = \gamma (v_2) - \gamma (v_1) \\
    &= \min_{\sigma \in \Sigma(\T, \{i\}, v_2)} \min_{\zeta \in Z^\sim(\T,\{i\},v_2)} \ell(\epsilon | v_2,\sigma, \zeta) - \min_{\sigma \in \Sigma(\T, \{i\}, v_1)} \min_{\zeta \in Z^\sim(\T,\{i\},v_1)} \ell(\epsilon | v_1,\sigma, \zeta) \\
    &= 0 - 0 = 0
\end{align*}

\begin{align*}
    \rlike (\{ i\} , i , v_1, \mathtt{test})&= \Delta \gamma (v_1, \mathtt{evade}) = \gamma (v_3) - \gamma (v_1) \\
    &= \min_{\sigma \in \Sigma(\T, \{i\}, v_3)} \min_{\zeta \in Z^\sim(\T,\{i\},v_3)} \ell(\epsilon | v_3,\sigma, \zeta) - \min_{\sigma \in \Sigma(\T, \{i\}, v_1)} \min_{\zeta \in Z^\sim(\T,\{i\},v_1)} \ell(\epsilon | v_1,\sigma, \zeta) \\
    &= 0-0 = 0
\end{align*}

\begin{align*}
    \rlike (\{ i,j\} , i , v_1, \mathtt{ignore})&= \Delta \gamma (v_1, \mathtt{ignore}) \\
    &= \gamma (v_2) - \gamma (v_1) \\
    &= \min_{\sigma \in \Sigma(\T, \{i,j\}, v_2)} \min_{\zeta \in Z^\sim(\T,\{i,j\},v_2)} \ell(\epsilon | v_2,\sigma, \zeta) \\
    & \phantom{\min_{\sigma \in \Sigma(\T, \{i,j\}, v_2)}} - \min_{\sigma \in \Sigma(\T, \{i,j\}, v_1)} \min_{\zeta \in Z^\sim(\T,\{i,j\},v_1)} \ell(\epsilon | v_1,\sigma, \zeta) \\
    &= 0 - 0 = 0
\end{align*}

\begin{align*}
    \rlike (\{ i,j\} , i , v_1, \mathtt{test})&= \Delta \gamma (v_1, \mathtt{evade}) \\
    &= \gamma (v_3) - \gamma (v_1) \\
    &= \min_{\sigma \in \Sigma(\T, \{i,j\}, v_3)} \min_{\zeta \in Z^\sim(\T,\{i,j\},v_3)} \ell(\epsilon | v_3,\sigma, \zeta) \\
    & \phantom{\min_{\sigma \in \Sigma(\T, \{i,j\}, v_2)}} - \min_{\sigma \in \Sigma(\T, \{i,j\}, v_1)} \min_{\zeta \in Z^\sim(\T,\{i,j\},v_1)} \ell(\epsilon | v_1,\sigma, \zeta) \\
    &= 0-0 = 0
\end{align*}

\begin{align*}
    \rrisk (\{ i\} , i , v_1, \mathtt{ignore})&= \max_{\zeta \in Z^\sim(\T, \{i\}, v_1)} \Delta \omega(v_1, \zeta, \mathtt{ignore}) = \max_{\zeta \in Z^\sim(\T, \{i\}, v_1)} [\omega (v_2, \zeta) - \omega(v_1, \zeta)] \\
    &= \max_{\zeta \in Z^\sim(\T, \{i\}, v_1)} [\min_{\sigma\in\Sigma(\T, \{i\}, v_2)} \ell(\epsilon | v_2,\sigma,\zeta) - \min_{\sigma\in\Sigma(\T, \{i\}, v_1)} \ell(\epsilon | v_1,\sigma,\zeta)]\\
    &= \max_{\zeta \in \{\mathtt{pass, test}\}} [\min_{\sigma\in\Sigma(\T, \{i\}, v_2)} \ell(\epsilon | v_2,\sigma,\zeta) - \min_{\sigma\in\Sigma(\T, \{i\}, v_1)} \ell(\epsilon | v_1,\sigma,\zeta)]\\
    &= \max\{ [1- \min_{\sigma\in\Sigma(\T, \{i\}, v_1)} \ell(\epsilon | v_1,\sigma,\mathtt{ignore})],  \\
    & \ [0 - \min_{\sigma\in\Sigma(\T, \{i\}, v_1)} \ell(\epsilon | v_1,\sigma,\mathtt{test})]\} \\
    &= \max \{ 1 - 0, 0 - 0 \} = 1 
\end{align*}

\begin{align*}
    \rrisk (\{ i\} , i , v_1, \mathtt{test})&= \max_{\zeta \in Z^\sim(\T, \{i\}, v_1)} \Delta \omega(v_1, \zeta, \mathtt{test}) = \max_{\zeta \in Z^\sim(\T, \{i\}, v_1)} [\omega (v_3, \zeta) - \omega(v_1, \zeta)] \\
    &= \max_{\zeta \in Z^\sim(\T, \{i\}, v_1)} [\min_{\sigma\in\Sigma(\T, \{i\}, v_3)} \ell(\epsilon | v_3,\sigma,\zeta) - \min_{\sigma\in\Sigma(\T, \{i\}, v_1)} \ell(\epsilon | v_1,\sigma,\zeta)]\\
    &= \max_{\zeta \in \{\mathtt{pass, test}\}} [\min_{\sigma\in\Sigma(\T, \{i\}, v_3)} \ell(\epsilon | v_3,\sigma,\zeta) - \min_{\sigma\in\Sigma(\T, \{i\}, v_1)} \ell(\epsilon | v_1,\sigma,\zeta)]\\
    &= \max\{ [0 - \min_{\sigma\in\Sigma(\T, \{i\}, v_1)} \ell(\epsilon | v_1,\sigma,\mathtt{ignore})], [0 - \min_{\sigma\in\Sigma(\T, \{i\}, v_1)} \ell(\epsilon | v_1,\sigma,\mathtt{test})]\} \\
    &= \max \{ 0 - 0, 0-0 \} = 0
\end{align*}

\begin{align*}
    \rrisk (\{ i,j\} , i , v_1, \mathtt{ignore})&= \max_{\zeta \in Z^\sim(\T, \{i, j\}, v_1)} \Delta \omega(v_1, \zeta, \mathtt{ignore}) = \max_{\zeta \in Z^\sim(\T, \{i,j\}, v_1)} [\omega (v_2, \zeta) - \omega(v_1, \zeta)] \\
    &= \max_{\zeta \in Z^\sim(\T, \{i,j\}, v_1)} [\min_{\sigma\in\Sigma(\T, \{i,j\}, v_2)} \ell(\epsilon | v_2,\sigma,\zeta) - \min_{\sigma\in\Sigma(\T, \{i,j\}, v_1)} \ell(\epsilon | v_1,\sigma,\zeta)]\\
    &= 0 - 0 = 0
\end{align*}

\begin{align*}
    \rrisk (\{ i,j\} , i , v_1, \mathtt{test})&= \max_{\zeta \in Z^\sim(\T, \{i, j\}, v_1)} \Delta \omega(v_1, \zeta, \mathtt{test}) = \max_{\zeta \in Z^\sim(\T, \{i,j\}, v_1)} [\omega (v_3, \zeta) - \omega(v_1, \zeta)] \\
    &= \max_{\zeta \in Z^\sim(\T, \{i,j\}, v_1)} [\min_{\sigma\in\Sigma(\T, \{i,j\}, v_3)} \ell(\epsilon | v_3,\sigma,\zeta) - \min_{\sigma\in\Sigma(\T, \{i,j\}, v_1)} \ell(\epsilon | v_1,\sigma,\zeta)]\\
    &= 0 - 0 = 0
\end{align*}

$\underline\rho(\{i\}, i, v_1) = \min\limits_{\a \in A_{v_1}} \rrisk(\{i\}, i, v_1, \a) = 0 $, hence

\begin{align*}
    \rneg(\{i\}, i, v_1, \mathtt{ignore}) &=  \max_{\zeta \in Z^\sim(\T, \{i\}, v_1)} \Delta \omega(v_1, \zeta, \mathtt{ignore}) - \underline\rho(v_1) \\
    &= \rrisk (\{ i\} , i , v_1, \mathtt{ignore}) - \underline\rho(v_1) = 1 - 0 = 1
\end{align*}

\begin{align*}
    \rneg(\{i\}, i, v_1, \mathtt{test}) &=  \max_{\zeta \in Z^\sim(\T, \{i\}, v_1)} \Delta \omega(v_1, \zeta, \mathtt{test}) - \underline\rho(v_1) \\
    &= \rrisk (\{ i\} , i , v_1, \mathtt{test}) - \underline\rho(v_1) = 0 - 0 = 0
\end{align*}

And $\underline\rho(\{i,j\}, i ,v_1) = \min\limits_{\a \in A_{v_1}} \rrisk(\{i, j\}, i, v_1, \a) = 0 $, hence

\begin{align*}
    \rneg(\{i,j\}, i, v_1, \mathtt{ignore}) &=  \max_{\zeta \in Z^\sim(\T, \{i,j\}, v_1)} \Delta \omega(v_1, \zeta, \mathtt{ignore}) - \underline\rho(v_1) \\
    &= \rrisk (\{ i,j\} , i , v_1, \mathtt{ignore}) - \underline\rho(\{i,j\}, i, v_1) = 0 - 0 = 0
\end{align*}

\begin{align*}
    \rneg(\{i,j\}, i, v_1, \mathtt{test}) &=  \max_{\zeta \in Z^\sim(\T, \{i, j\}, v_1)} \Delta \omega(v_1, \zeta, \mathtt{test}) - \underline\rho(\{i,j\}, i, v_1) \\
    &= \rrisk (\{ i,j\} , i , v_1, \mathtt{test}) - \underline\rho(\{i,j\}, i, v_1) = 0 - 0 = 0
\end{align*}

\paragraph{$\T$ as in \Cref{fig:coordination_rescue}:}

\begin{align*}
    \rlike (\{ i\} , i , v_1, \mathtt{left})&= \Delta \gamma (v_1, \mathtt{left}) = \gamma (v_2) - \gamma (v_1) \\
    &= \min_{\sigma \in \Sigma(\T, \{i\}, v_2)} \min_{\zeta \in Z^\sim(\T,\{i\},v_2)} \ell(\epsilon | v_2,\sigma, \zeta) - \min_{\sigma \in \Sigma(\T, \{i\}, v_1)} \min_{\zeta \in Z^\sim(\T,\{i\},v_1)} \ell(\epsilon | v_1,\sigma, \zeta) \\
    &= 0 - 0 = 0
\end{align*}

\begin{align*}
    \rlike (\{ i\} , i , v_1, \mathtt{right})&= \Delta \gamma (v_1, \mathtt{right}) = \gamma (v_3) - \gamma (v_1) \\
    &= \min_{\sigma \in \Sigma(\T, \{i\}, v_3)} \min_{\zeta \in Z^\sim(\T,\{i\},v_3)} \ell(\epsilon | v_3,\sigma, \zeta) - \min_{\sigma \in \Sigma(\T, \{i\}, v_1)} \min_{\zeta \in Z^\sim(\T,\{i\},v_1)} \ell(\epsilon | v_1,\sigma, \zeta) \\
    &= 0-0 = 0
\end{align*}

\begin{align*}
    \rlike (\{ i,j\} , i , v_1, \mathtt{left})&= \Delta \gamma (v_1, \mathtt{left}) \\
    &= \gamma (v_2) - \gamma (v_1) \\
    &= \min_{\sigma \in \Sigma(\T, \{i,j\}, v_2)} \min_{\zeta \in Z^\sim(\T,\{i,j\},v_2)} \ell(\epsilon | v_2,\sigma, \zeta) \\
    & \phantom{\min_{\sigma \in \Sigma(\T, \{i,j\}, v_2)}} - \min_{\sigma \in \Sigma(\T, \{i,j\}, v_1)} \min_{\zeta \in Z^\sim(\T,\{i,j\},v_1)} \ell(\epsilon | v_1,\sigma, \zeta) \\
    &= 0 - 0 = 0
\end{align*}

\begin{align*}
    \rlike (\{ i,j\} , i , v_1, \mathtt{right})&= \Delta \gamma (v_1, \mathtt{right}) \\
    &= \gamma (v_3) - \gamma (v_1) \\
    &= \min_{\sigma \in \Sigma(\T, \{i,j\}, v_3)} \min_{\zeta \in Z^\sim(\T,\{i,j\},v_3)} \ell(\epsilon | v_3,\sigma, \zeta) \\
    & \phantom{\min_{\sigma \in \Sigma(\T, \{i,j\}, v_2)}} - \min_{\sigma \in \Sigma(\T, \{i,j\}, v_1)} \min_{\zeta \in Z^\sim(\T,\{i,j\},v_1)} \ell(\epsilon | v_1,\sigma, \zeta) \\
    &= 0-0 = 0
\end{align*}

\begin{align*}
    \rrisk (\{ i\} , i , v_1, \mathtt{left})&= \max_{\zeta \in Z^\sim(\T, \{i\}, v_1)} \Delta \omega(v_1, \zeta, \mathtt{left}) = \max_{\zeta \in Z^\sim(\T, \{i\}, v_1)} [\omega (v_2, \zeta) - \omega(v_1, \zeta)] \\
    &= \max_{\zeta \in Z^\sim(\T, \{i\}, v_1)} [\min_{\sigma\in\Sigma(\T, \{i\}, v_2)} \ell(\epsilon | v_2,\sigma,\zeta) - \min_{\sigma\in\Sigma(\T, \{i\}, v_1)} \ell(\epsilon | v_1,\sigma,\zeta)]\\
    &= \max_{\zeta \in \{\mathtt{left, right}\}} [\min_{\sigma\in\Sigma(\T, \{i\}, v_2)} \ell(\epsilon | v_2,\sigma,\zeta) - \min_{\sigma\in\Sigma(\T, \{i\}, v_1)} \ell(\epsilon | v_1,\sigma,\zeta)]\\
    &= \max\{ [0 - \min_{\sigma\in\Sigma(\T, \{i\}, v_1)} \ell(\epsilon | v_1,\sigma,\mathtt{left})],  \\
    & \ [1 - \min_{\sigma\in\Sigma(\T, \{i\}, v_1)} \ell(\epsilon | v_1,\sigma,\mathtt{right})]\} \\
    &= \max \{ 0 - 0, 1 - 0 \} = 1 
\end{align*}

\begin{align*}
    \rrisk (\{ i\} , i , v_1, \mathtt{right})&= \max_{\zeta \in Z^\sim(\T, \{i\}, v_1)} \Delta \omega(v_1, \zeta, \mathtt{right}) = \max_{\zeta \in Z^\sim(\T, \{i\}, v_1)} [\omega (v_3, \zeta) - \omega(v_1, \zeta)] \\
    &= \max_{\zeta \in Z^\sim(\T, \{i\}, v_1)} [\min_{\sigma\in\Sigma(\T, \{i\}, v_3)} \ell(\epsilon | v_3,\sigma,\zeta) - \min_{\sigma\in\Sigma(\T, \{i\}, v_1)} \ell(\epsilon | v_1,\sigma,\zeta)]\\
    &= \max_{\zeta \in \{\mathtt{left, right}\}} [\min_{\sigma\in\Sigma(\T, \{i\}, v_3)} \ell(\epsilon | v_3,\sigma,\zeta) - \min_{\sigma\in\Sigma(\T, \{i\}, v_1)} \ell(\epsilon | v_1,\sigma,\zeta)]\\
    &= \max\{ [1 - \min_{\sigma\in\Sigma(\T, \{i\}, v_1)} \ell(\epsilon | v_1,\sigma,\mathtt{left})], [0 - \min_{\sigma\in\Sigma(\T, \{i\}, v_1)} \ell(\epsilon | v_1,\sigma,\mathtt{right})]\} \\
    &= \max \{ 1 - 0, 0-0 \} = 1
\end{align*}

\begin{align*}
    \rrisk (\{ i,j\} , i , v_1, \mathtt{left})&= \max_{\zeta \in Z^\sim(\T, \{i, j\}, v_1)} \Delta \omega(v_1, \zeta, \mathtt{left}) = \max_{\zeta \in Z^\sim(\T, \{i,j\}, v_1)} [\omega (v_2, \zeta) - \omega(v_1, \zeta)] \\
    &= \max_{\zeta \in Z^\sim(\T, \{i,j\}, v_1)} [\min_{\sigma\in\Sigma(\T, \{i,j\}, v_2)} \ell(\epsilon | v_2,\sigma,\zeta) - \min_{\sigma\in\Sigma(\T, \{i,j\}, v_1)} \ell(\epsilon | v_1,\sigma,\zeta)]\\
    &= 0 - 0 = 0
\end{align*}

\begin{align*}
    \rrisk (\{ i,j\} , i , v_1, \mathtt{right})&= \max_{\zeta \in Z^\sim(\T, \{i, j\}, v_1)} \Delta \omega(v_1, \zeta, \mathtt{right}) = \max_{\zeta \in Z^\sim(\T, \{i,j\}, v_1)} [\omega (v_3, \zeta) - \omega(v_1, \zeta)] \\
    &= \max_{\zeta \in Z^\sim(\T, \{i,j\}, v_1)} [\min_{\sigma\in\Sigma(\T, \{i,j\}, v_3)} \ell(\epsilon | v_3,\sigma,\zeta) - \min_{\sigma\in\Sigma(\T, \{i,j\}, v_1)} \ell(\epsilon | v_1,\sigma,\zeta)]\\
    &= 0 - 0 = 0
\end{align*}

$\underline\rho(\{i\}, i, v_1) = \min\limits_{\a \in A_{v_1}} \rrisk(\{i\}, i, v_1, \a) = 1 $, hence

\begin{align*}
    \rneg(\{i\}, i, v_1, \mathtt{left}) &=  \max_{\zeta \in Z^\sim(\T, \{i\}, v_1)} \Delta \omega(v_1, \zeta, \mathtt{left}) - \underline\rho(\{i\}, i, v_1) \\
    &= \rrisk (\{ i\} , i , v_1, \mathtt{left}) - \underline\rho(\{i\}, i, v_1) = 1 - 1 = 0
\end{align*}

\begin{align*}
    \rneg(\{i\}, i, v_1, \mathtt{right}) &=  \max_{\zeta \in Z^\sim(\T, \{i\}, v_1)} \Delta \omega(v_1, \zeta, \mathtt{right}) - \underline\rho(\{i\}, i, v_1) \\
    &= \rrisk (\{ i\} , i , v_1, \mathtt{right}) - \underline\rho(\{i\}, i, v_1) = 1 - 1 = 0
\end{align*}

And $\underline\rho(\{i,j\}, i ,v_1) = \min\limits_{\a \in A_{v_1}} \rrisk(\{i, j\}, i, v_1, \a) = 0 $, hence

\begin{align*}
    \rneg(\{i,j\}, i, v_1, \mathtt{left}) &=  \max_{\zeta \in Z^\sim(\T, \{i,j\}, v_1)} \Delta \omega(v_1, \zeta, \mathtt{left}) - \underline\rho(\{i,j\}, i, v_1) \\
    &= \rrisk (\{ i,j\} , i , v_1, \mathtt{left}) - \underline\rho(\{i,j\}, i, v_1) = 0 - 0 = 0
\end{align*}

\begin{align*}
    \rneg(\{i,j\}, i, v_1, \mathtt{right}) &=  \max_{\zeta \in Z^\sim(\T, \{i, j\}, v_1)} \Delta \omega(v_1, \zeta, \mathtt{right}) - \underline\rho(\{i,j\}, i, v_1) \\
    &= \rrisk (\{ i,j\} , i , v_1, \mathtt{right}) - \underline\rho(\{i,j\}, i, v_1) = 0 - 0 = 0
\end{align*}

\section{Axiom Compliance for Member Contribution Functions}\label{app:pw_proofs}

Here we present the proof omitted in \Cref{sec:pwresp}.

\pwrespaxioms*
\begin{proof}
We consider the axioms one by one.
For (KSym) first note that $\rrisk$ and $\rneg$ fulfill this axiom by definition. The avoidance $\omega$ as well as the minimal risk $\underline{\rho}$ are defined over the whole information set of a given node and both member contribution functions then maximise over $Z^\sim(v)$, which also takes into account all vertices in the same information set as $v$. 
For $\rlike$ this is different. Even though the known guaranteed likelihood $\gamma(v)$ is defined over $\Sigma(v)$ and $Z^\sim(v)$, the member contribution function then measures the difference in $\gamma$ between the node $c_v(\a)$, the agent ends up in after taking the action $\a$, and $v$, the node the agent was before taking action $\a$. Since $c_v(\a)$ depends on the actual node $v$ -- and not its information set -- we cannot guarantee (KSym). For an actual example where (KSym) is violated consider the coordination game depicted in \Cref{fig:coordination_rescue}. Here we have 
\[\rlike(j, v_1, \mathtt{cinema}) = 0 \neq 1 = \rlike(j, v_2, \mathtt{cinema}). \]

It is clear from definition that $\rlike$ satisfy (AMC) as it measures the actual increase in guaranteed known likelihood which is 0 in the setting of (AMC). 
% To prove compliance for $\r^2$ consider $\bar{\mu}(v) = \min_{\sigma \in \Sigma(v)\setminus \a} \max_{\zeta \in Z^\sim(v)} \ell(\epsilon) \geq \ell(\epsilon \mid v, \zeta', \b) = 1$ where $\zeta'$ is a scenario where agent $i$ is actually in node $v$. Thus w.l.o.g.\ we can assume that for the calculation of $\mu(v)$ a strategy $\sigma$ where $\a$ is the first action is chosen. Then we have $\r^2(v, \a) = \max_{v \sim v} \mu(c_v(\a)) - \mu(v) = 0$. 
To prove compliance of $\rneg$ with (AMC) assume we are in a situation as described in the definition of the axiom, i.e., a node $v \in V_i$ is given for an agent $i$ and $A_v = \{ \a, \b_1, \ldots, \b_m\}$ such that choosing $\a$ certainly leads to a desirable outcome and choosing any other action certainly leads to an undesirable outcome. We first note that 
\[ \min_{\b \in A_{v}\setminus \{\a\}} \rho(v, \b) = \min_{b \in A_{v}\setminus \{\a\}} \max_{\zeta \in Z^\sim} \omega (c_{v}(\b) \zeta) - \omega(v, \zeta) \geq \omega(c_{v}(\b), \zeta') -\omega(v, \zeta') = 1 - 0 = 1,\]
where $\zeta'$ is a scenario where agent $i$ is actually in node $v$. Thus w.l.o.g.\ we can assume that for computation of $\underline\rho$ action $\a$ is used. Then we have $\rneg(i, v, \a) = \rho(v, \a) - \rho(v, \a) = 0$, which proves compliance. 
$\rrisk$ however does not satisfy (AMC). To see this consider again the coordination game depicted in \Cref{fig:coordination_rescue}. Here, 
\[\rrisk(j, v_1, \mathtt{cinema}) = \max_{\zeta \in Z^\sim(v_1)} \omega(c_{v_1}(\mathtt{cinema}), \zeta) - \omega(v_1, \zeta) = 1.\]
(AMC$^\sim$) is by definition satisfied by all functions that satisfy (AMC). Additionally, also $\rrisk$ satisfies it since choosing action $\a$ is the best action in all nodes $v'\sim v$.

(FMC) is again clearly satisfied by $\rlike$. $\rneg$ does not satisfy it, however. To see this consider the same coordination game as above where $\rneg(j, v_1, \mathtt{theater}) = 0$. 
$\rrisk$ on the other hand satisfies the axiom. To see this assume we are in a situation as described in the definition of the axiom, i.e., a node $v \in V_i$ is given for an agent $i$ and $A_v = \{ \a, \b_1, \ldots, \b_m\}$ such that choosing $\a$ certainly leads to an undesirable outcome and choosing any other action certainly leads to a desirable outcome. Let $\zeta'$ be the scenario where agent $i$ is actually in node $v$ again. Then 
\[\rrisk(i, v, \a) = \max_{\zeta \in Z^\sim(v)} \omega(c_{v}(\a), \zeta) - \omega(v, \zeta) \leq \omega(c_{v}(\a), \zeta') - \omega(v, \zeta) = 1 - 0 = 1. \]
As was the case with (AMC$^\sim$), for (FMC$^\sim$) we only have to check the function that does not satisfy the stronger version of the axiom. Here, this is $\rneg$. For this member contribution function we have $\rneg(i, v, \a) = \rrisk(i, v, \a) - \min_{\b \in A} \rrisk(v, \b) = 1 - 0 = 1$.
\end{proof}

\section{Axiom Compliance for Aggregation Functions}\label{sec:appendix-evaluation-agg}

Here we present the proofs omitted in \Cref{sec:agg}. For each aggregation function we prove the stated axiom compliance resulting in \Cref{prop:agg_compliance}.

\begin{lemma}
$\mathbf{sum}$ fulfills all of the given axioms except (AN1) and (SIP). Further, it fails (01B) and is thus not a proper aggregation function.
\end{lemma}

\begin{proof}
We provide arguments for each axiom individually.
\begin{description}
    % \item[]
    \item[(01B)] It is clear by definition that this is not satisfied..
    % \item[(1ID)] It holds that $\sum_{i=1}^1 x = x$ for all $x \in [0,1]$.
    \item[(BSM$^+$)] Let $\ov{\r} = (\r_1,\ldots,\r_n) \in [0,1]^n$  and $\ov{\r}' = (\r_1,\ldots,\r_n, \r_{n+1})\in [0,1]^{n+1} $ be point-wise responsibilities, s.t. $\mathbf{sum}(\ov{\r})<1$ and $ \r_{n+1}>0$. Then
    \[ \mathbf{sum}(\ov{\r}') = \sum_{k=1}^{n+1} \r_k = \sum_{k=1}^{n} \r_k + \r_{n+1} = \mathbf{sum}(\ov{\r}) + \r_{n+1} > \mathbf{sum}(\ov{\r}). \]
    \item[(BSM$^>$)] Let $\ov{\r} = (\r_1,\ldots,\r_n) \in [0,1]^n $  and $\ov{\r}' = (\r_1,\ldots,\r_{m-1}, \r'_m,\ldots, \r_n) \in [0,1]^n $ with $m\leq n$ be point-wise responsibilities such that $\mathbf{sum}(\ov{\r})<1$ and  $\r'_m > \r_m$. Then
    \[ \mathbf{sum} (\ov{r}') = \sum_{k=1}^{m-1 } \r_k + \r'_m + \sum_{k=m+1}^n \r_k = \mathbf{sum}(\ov{\r}) + (\r'_m - \r_m) > \mathbf{sum}(\ov{\r}). \]
    \item[(LIN)] Let $n\in \mathbb N _{\geq 1}$, $\overline{\r} = (\r_1, \ldots, \r_n) \in [0,1]^n$ point-wise responsibilities and $i\leq n$ an index. Then
    \[ \mathbf{sum}^{(n)}_{\overline{\r}, -i}(x) = \mathbf{sum} ((\r_1, \ldots, \r_{i-1}, x, \r_i, \ldots, r_n)) =  \underbrace{\sum\limits_{\substack {k=1 \\ k\neq i}}^{n} r_k}_{=:b} + \underbrace{1}_{=:a}x = b+ax.  \]
    \item[(AN1)] This is not satisfied since $\mathbf{sum}((1,1)) = 2 \neq 1$.
    \item[(NE0)] Clear by definition.
    \item[(SIP)] This is not satisfied since $\mathbf{sum}((1,1)) = 2 \neq 1$.
    \item[(AAT)] Holds due to commutativity of the sum.
    \item[(RED)] For the first direction observe that $\sum_{k=1}^n 0 = 0$.
    For the other direction let $\ov{\r}=(\r_1, \ldots, \r_n) \in [0,1]^n $ be point-wise responsibilities such that $ \sum_{k=1}^n \r_k = 0$. Then
    \begin{align*}
    \sum_{k=1}^n \r_k &= 0\text{ and }\r_k \in [0,1] \text{ for all }k=1, \ldots, n \\
    \Rightarrow \ \r_k &= 0\text{ for all } k=1, \ldots, n \\
    \Rightarrow \ \ov{\r}&= (0,\ldots,0).
    \end{align*}
    \end{description}
\end{proof}

\begin{lemma}
$\mathbf{avg}$ fulfills (01B), (BSM$^>$), (AAT) and (RED), but does not fulfill (BSM$^+$), (LIN) and (NE0).
\end{lemma}

\begin{proof}
We provide arguments for each axiom individually.
\begin{description}
    % \item[]
    \item[(01B)] Let $\ov{\r} = (\r_1,\ldots,\r_n) \in [0,1]^n$, i.e. $r_k \in [0,1]$ for all $k=1,\ldots,n$. Then $\frac{1}{n}\sum_{k=1}^n \r_k \in [0,1]$. 
    % \item[(1ID)] It holds that $\frac{1}{1}\sum_{i=1}^1 x = x$ for all $x \in [0,1]$.
    \item[(BSM$^+$)] To see that this is violated let $\ov{\r} = (0.9) \in [0,1]^1$  and $\ov{\r}' = (0.9,0.1)\in [0,1]^2 $ be point-wise responsibilities. $0.1>0$. Then
    \[ \mathbf{avg}(\ov{\r}') = \frac{0.9 + 0.1}{2} = 0.5 < 0.9 =  \mathbf{avg}(\ov{\r}). \]
    \item[(BSM$^>$)] Let $\ov{\r} = (\r_1,\ldots,\r_n) \in [0,1]^n $  and $\ov{\r}' = (\r_1,\ldots,\r_{m-1}, \r'_m,\ldots, \r_n) \in [0,1]^n $ with $m\leq n$ be point-wise responsibilities such that $\mathbf{avg}(\ov{\r})<1$ and  $\r'_m > \r_m$. Then
    \[ \mathbf{avg} (\ov{r}') = \frac{1}{n}[\sum_{k=1}^{m-1 } \r_k + \r'_m + \sum_{k=m+1}^n \r_k] = \mathbf{avg}(\ov{\r}) + \frac{1}{n}(\r'_m - \r_m) > \mathbf{avg}(\ov{\r}). \]
    \item[(LIN)] To see why this is violated let $n\in \mathbb N _{\geq 1}$, $\overline{\r} = (\r_1, \ldots, \r_n) \in [0,1]^n$ point-wise responsibilities and $i\leq n$ an index. Then
    \[ \mathbf{avg}^{(n)}_{\overline{\r}, -i}(x) = \mathbf{avg} ((\r_1, \ldots, \r_{i-1}, x, \r_i, \ldots, r_n)) =  \underbrace{\frac{1}{n}\sum\limits_{\substack {k=1 \\ k\neq i}}^{n} r_k}_{=:b} + \underbrace{\frac{1}{n}}_{=:a}x = b+ax. \]
    \item[(AN1)] It holds that $\mathbf{avg}((1,0.2)) = \frac{1}{2} \times (1+0.2) = 0.6 \neq 1$.
    \item[(NE0)] Clear by definition that this is not satisfied.
    \item[(SIP)] Clear by definition.
    \item[(AAT)] Holds due to commutativity of the sum.
    \item[(RED)] For one direction observe that $\frac{\sum_{k=1}^n 0}{n} = 0$.
    For the other direction let $\ov{\r}=(\r_1, \ldots, \r_n) \in [0,1]^n $ be point-wise responsibilities such that $\frac{1}{n} \sum_{k=1}^n \r_k = 0$. Then
    \begin{align*}
    && \frac{1}{n}\sum_{k=1}^n \r_k  &= 0 \text{ and }\r_k \in [0,1] \text{ for all }k=1, \ldots, n\\
    \Rightarrow&&  \sum_{k=1}^n \r_k &= 0\text{ and }\r_k \in [0,1] \text{ for all }k=1, \ldots, n \\
    \Rightarrow&&  \r_k &= 0\text{ for all } k=1, \ldots, n \\
    \Rightarrow&& \ov{\r} &= (0,\ldots,0).
    \end{align*}
    \end{description}
\end{proof}

\begin{lemma}
$\mathbf{max}$ fulfills (01B), (1ID), (AAT) and (RED), but does not fulfill (BSM$^+$) nor (BSM$^>$), (LIN), and (NE0).
\end{lemma}

\begin{proof}
We provide arguments for each axiom individually.
\begin{description}
    % \item[]
    \item[(01B)] Clear by definition.    
    % \item[(1ID)] It holds that $\max (x) = x$ for all $x \in [0,1]$.
    \item[(BSM$^+$)] To see why this is violated let $\ov{\r} = (0.5) \in [0,1]^1$  and $\ov{\r}' = (0.5,0.4)\in [0,1]^2 $ be point-wise responsibilities. $0.4>0$. Then
    \[ \mathbf{max}(\ov{\r}') = 0.5 = \mathbf{max}(\ov{\r}). \]
    \item[(BSM$^>$)] To see why this is violated let $\ov{\r} = (0.5,0.8) \in [0,1]^2 $  and $\ov{\r}' = (0.6, 0.8) \in [0,1]^2 $. $0.6 > 0.5$. Then
    \[\mathbf{max} (\ov{r}') = 0.8 = \mathbf{max} (\ov{r}). \]
    \item[(LIN)] To see why this is violated let $n=2$, $\overline{\r} = (0.5, 0.5) \in [0,1]^2$ and $i=1$. Then
    \[ \mathbf{max}^{(n)}_{\overline{\r}, -i}(x) = \mathbf{max}((x, 0.5)) = \max \{x, 0.5 \} = \begin{cases} x \qquad \text{ if }x\geq 0.5 \\ 0.5 \qquad \text{ else.}\end{cases}\] 
    This is not linear.
    \item[(AN1)] Clear by definition.
    \item[(NE0)] Clear by definition.
    \item[(SIP)] Clear by definition.
    \item[(AAT)] Holds due to commutativity of the maximum.
    \item[(RED)] For one direction observe that $\max (0, \ldots, 0) = 0$.
    For the other direction let  $\ov{\r}=(\r_1, \ldots, \r_n) \in [0,1]^n $ be point-wise responsibilities such that $\max(\ov{\r}) = 0$. Assume $\r_k > 0$ for some $k\in\{1,\ldots,n\}$. Then $\max(\ov{\r}) = r_k > 0 $. Contradiction.
    \end{description}
\end{proof}

\begin{lemma}
$\mathbf{mprod}$ fulfills all of the mentioned axioms except for (SIP).
\end{lemma}

\begin{proof}
We provide arguments for each axiom individually.
\begin{description}
    % \item[]
    \item[(01B)] Let $\ov{\r} = (\r_1,\ldots,\r_n) \in [0,1]^n$ be point-wise responsibilities. $\r_k \in [0,1]$ for all $k=1,\ldots,n$, hence $1-\r_k \in [0,1]$ for all $k=1,\ldots,n$, $\prod_{k=1}^n (1-\r_k) \in [0,1]$, and finally $\mathbf{mprod}(\ov{\r})= 1- [\prod_{k=1}^n (1-\r_k) ]\in [0,1] $.
    % \item[(1ID)] It holds that $1- [\prod_{i=1}^1 (1-  x))] = 1-[1-x] = x $ for all $x\in[0,1]$.
    \item[(BSM$^+$)] Let $\ov{\r} = (\r_1,\ldots,\r_n) \in [0,1]^n$  and $\ov{\r}' = (\r_1,\ldots,\r_n, \r_{n+1})\in [0,1]^{n+1} $ be point-wise responsibilities, s.t. $\mathbf{mprod}(\ov{\r})<1$ and $\r_{n+1}>0$. Then
    \[ \mathbf{mprod}(\ov{\r}') = 1- [\prod_{k=1}^{n+1} (1-\r_k)] = 1- [\prod_{k=1}^{n} (1-\r_k) \times (1-\r_{n+1})] . \]
    It holds that $\r_{n+1} \in (0,1] $, so $(1-\r_{n+1}) \in [0,1)$, and $\prod_{k=1}^{n} (1-\r_k) \times (1-\r_{n+1}) < \prod_{k=1}^{n} (1-\r_k)$. Therefore $\mathbf{mprod}(\ov{\r}') = 1- [\prod_{k=1}^{n} (1-\r_k) \times (1-\r_{n+1})] > 1- [\prod_{k=1}^{n} (1-\r_k)] = \mathbf{mprod}(\ov{\r})$.
    
    \item[(BSM$^>$)] Let $\ov{\r} = (\r_1,\ldots,\r_n) \in [0,1]^n $  and $\ov{\r}' = (\r_1,\ldots,\r_{m-1}, \r'_m,\ldots, \r_n) \in [0,1]^n $ with $m\leq n$ be point-wise responsibilities such that $\mathbf{mprod}(\ov{\r})<1$ and $\r'_m > \r_m$. Then
    \[ \mathbf{mprod} (\ov{r}') = 1- [\prod_{k=1}^{m-1} (1-\r_k) \times (1-r'_m) \times \prod_{k=m+1}^n \r_k] = 1- [\prod_{k=1}^{m-1} (1-\r_k) \times \prod_{k=m+1}^n \r_k \times (1-r'_m)] .\]
    It holds that $\r'_m > \r_m$ and $\r_m, \r'_m \in [0,1]$, hence $1-\r'_m < 1-\r_m$, and $ \prod_{k=1}^{m-1} (1-\r_k) \times \prod_{k=m+1}^n \r_k \times (1-r'_m) < \prod_{k=1}^{m-1} (1-\r_k) \times \prod_{k=m+1}^n \r_k \times (1-r_m)$. Thus $ \mathbf{mprod} (\ov{r}') = 1- [\prod_{k=1}^{m-1} (1-\r_k) \times \prod_{k=m+1}^n \r_k \times (1-r'_m)] > 1- [\prod_{k=1}^{m-1} (1-\r_k) \times \prod_{k=m+1}^n \r_k \times (1-r_m)] = \mathbf{mprod} (\ov{r}) $.
    \item[(LIN)] Let $n\in \mathbb N _{\geq 1}$, $\overline{\r} = (\r_1, \ldots, \r_n) \in [0,1]^n$ be given point-wise responsibilities and $i\leq n$ be an index. Then
    \begin{align*}
        \mathbf{mprod}^{(n)}_{\overline{\r}, -i}(x) &= \mathbf{mprod} ((\r_1, \ldots, \r_{i-1}, x, \r_i, \ldots, r_n)) \\ 
        &= 1- \prod\limits_{\substack {k=1 \\ k\neq i}}^{n} (1-r_k) \times (1-x) \\
        & = \underbrace{1- \prod\limits_{\substack {k=1 \\ k\neq i}}^{n} (1-r_k)}_{=: b}+\underbrace{ \prod\limits_{\substack {k=1 \\ k\neq i}}^{n} (1-r_k) }_{=:a} \times x \\
        & = b+ax.  
        \end{align*}

    \item[(AN1)] Let $\overline{r} = (\r_1, \ldots, \r_n) \in [0,1]^n$ be point-wise responsibilities such that $\r_i = 1$ for some $i\leq n$. Then $1-\r_i = 0$, hence $\prod_{k=1}^n (1-\r_k) = 0$ and $\mathbf{mprod}(\overline{\r}) = 1 - \prod_{k=1}^n (1-\r_k) = 1$.
    \item[(NE0)] Let $\overline{r} = (\r_1, \ldots, \r_n) \in [0,1]^n$ be point-wise responsibilities such that $\r_i = 0$ for some $i\leq n$. Then we have $1-\r_i = 1$ and thus $\prod\limits_{k=1}^n (1-\r_k) = \prod\limits_{\substack{k=1 \\ k\neq i}}^n (1-\r_k)$, giving us $\mathbf{mprod}((\r_1, \ldots, \r_n)) = \mathbf{mprod} ((\r_1, \ldots, \r_{i-1}, \r_{i+1}, \ldots, \r_n))$.
    \item[(SIP)] To see why this is violated observe that $1-(1-0.9)\times(1-0.9) = 1- 0.1\times0.1 = 1- 0.01 = 0.99 \neq 0.9$.
    \item[(AAT)] Holds due to commutativity of the product.
    \item[(RED)] For one direction observe that $1- [\prod_{k=1}^{n} (1-0)] = 1-[1] = 0$.
    For the other direction let $\ov{\r}=(\r_1, \ldots, \r_n) \in [0,1]^n $ be point-wise responsibilities such that $1- [\prod_{k=1}^{n} (1-\r_k)] = 0$. 
    \begin{align*}
        && 1- [\prod_{k=1}^{n} (1-\r_k)] &= 0\\
        \Leftrightarrow&& \qquad \prod_{k=1}^{n} (1-\r_k) &= 1 \\
        \xRightarrow[]{1-\r_k \in [0,1] \forall k} && \qquad 1-\r_k &= 1 \quad \text{ for all } k=1,\ldots, n\\
        \Leftrightarrow&& \qquad \qquad \r_k &= 0 \quad \text{ for all } k=1,\ldots, n.
    \end{align*}
    \end{description}
\end{proof}

\characterisationmprod*

For ease of readability we break up the proof into two lemmata.
We start out by showing that linearity in each component together with the other axioms prescribes a specific linear combination of the input values as aggregation function, before showing that this equates exactly to $\mathbf{mprod}$.

\begin{lemma}\label{lemma:lin-fnctn}
    The axioms (LIN), (AAT), (NE0), and (AN1) uniquely characterise the aggregation function 
    \[
    \agg([\r_1,\ldots,\r_n]) = 
    \sum\limits_{i=1}^n \r_i - \sum\limits_{\substack{i,i'= 1 \\ i<i'}}^n  \r_i \r_{i'} + \ldots \pm \sum\limits_{\substack{i^{(1)},\ldots, i^{(n-1)} = 1  \\ i^{(1)}<\ldots < i^{(n-1)}}}^n  \r_{i^{(1)}}\ldots\r_{i^{(n-1)}} \mp \r_1\ldots\r_n.
    \]
\end{lemma}

\begin{proof}
    The proof proceeds via induction on n. The property (1ID) of the definition of aggregation functions will play a crucial role and we will refer to it specifically. We explicitly name the step $n=2$ as it uses a similar reasoning to the general case of $n\to n+1$ and thus makes the latter step easier to follow.
    \begin{description}
        \item[\underline{n=1:}] $\agg([\r_1]) \stackrel{\text{(1ID)}}{=} \r_1. $ \checkmark
        \item[\underline{n=2:}] 
        \begin{align*}
         \agg([\r_1, \r_2]) \stackrel{\text{(LIN)}}{=} &f(\r_2) + g(\r_2) \r_1 = h(\r_1) + k(\r_1) \r_2 \\
        \stackrel{\text{(AAT)}}{\Rightarrow} \qquad  \agg([\r_1, \r_2]) \ = \ & f(\r_2) + g(\r_2) \r_1 = f(\r_1) + g(\r_1) \r_2. 
        \end{align*}
        For some functions $f, g, h$ and $k$.\\
        Now we can use the axioms (NE0) and (AN1) to solve these equations for specific values.
        \begin{align*}
            \agg([0, \r_2]) = f(\r_2) + g(\r_2)\cdot 0 = f(\r_2) \stackrel{\text{(NE0)}}{=} \agg([\r_2]) \stackrel{\text{(1ID)}}{=}\  &\r_2 \\
            \text{and} \qquad \agg([1, \r_2]) = f(\r_2) + g(\r_2)\cdot 1 = \r_2 + g(\r_2) \stackrel{\text{(AN1)}}{=} &1\\
            \Leftrightarrow \qquad g(\r_2) = \ & 1-\r_2.
        \end{align*}
        Thus \[\agg([\r_1, \r_2]) = f(\r_2) + g(\r_2) \r_1 = \r_2 + (1-\r_2) \r_1 = \r_2 + \r_1 - \r_1\r_2. \quad \text{\checkmark}\] 
        \item[\underline{n$\to$ n+1:}] Assume $\agg([\r_1, \ldots ,\r_n]) = \sum\limits_{i=1}^n \r_i - \sum\limits_{\substack{i,i'= 1 \\ i<i'}}^n  \r_i \r_{i'} + \ldots \pm \r_1\ldots\r_n.$ \\
        As before
        \[ \agg([\r_1,\ldots, \r_n, \r_{n+1}]) \stackrel{\text{(LIN)}}{=} f([\r_1, \ldots , \r_n]) + g([\r_1, \ldots, \r_n]) \r_{n+1} \]
        for some functions $f$ and $g$. Thus
        \[\agg([\r_1, \ldots, \r_n]) \stackrel{\text{(NE0)}}{=} \agg([\r_1, \ldots, \r_n, 0]) = f([\r_1, \ldots, \r_n]) + 0, \]
        but also
         \begin{alignat*}{3}
             \agg([\r_1, \ldots, \r_n, 1]) &&= f([\r_1, \ldots, \r_n]) + g([\r_1, \ldots, \r_n]) &\stackrel{\text{(AN1)}}{=}&& 1 \\
             \Leftrightarrow \qquad &&\agg([\r_1, \ldots, \r_n]) + g([\r_1, \ldots, \r_n]) &\ \ =&& 1\\
             && \Leftrightarrow \qquad g([\r_1, \ldots, \r_n]) &\ \ =&& 1 - \agg([\r_1, \ldots, \r_n]).
         \end{alignat*}
         Hence
         \begin{align*}
             &\agg([\r_1, \ldots, \r_{n+1}]) \\
             &= \agg([\r_1, \ldots, \r_n]) + \r_{n+1}[1 - \agg([\r_1, \ldots, \r_n])] \\
             &= \agg([\r_1, \ldots, \r_n]) + \r_{n+1} - \r_{n+1}\agg([\r_1, \ldots, \r_n])]\\
             &\stackrel{\mathclap{\text{ind.h.}}}{=} \,\,\, {\color{mycolor1}\sum\limits_{i=1}^{n}\r_i}  {\color{mycolor2} - \sum\limits_{\substack{i, i' = 1 \\ i< i'}}^n \r_i \r_{i'}} { \color{mycolor3} + \sum\limits_{\substack{ i,i',i''= 1 \\ i<i'< i''}}^n \r_i \r_{i'} \r_{i''}} - \ldots {\color{mycolor4}\pm \r_{i^{(1)}}\ldots \r_{i^{(n)}}}  {\color{mycolor1} + \r_{n+1}} \\
            &\phantom{=M} {\color{mycolor2} - \r_{n+1}\sum\limits_{i=1}^{n}\r_i } {\color{mycolor3} + \r_{n+1}\sum\limits_{\substack{i, i' = 1 \\ i< i'}}^n \r_i \r_{i'}}  - \ldots {\color{mycolor4} \pm \r_{n+1}\sum\limits_{\substack{i^{(1)},\ldots, i^{(n-1)}=1 \\ i^{(1)}<\ldots< i^{(n-1)}}}^n \r_{i^{(1)}} \ldots \r_{i^{(n-1)}}} \\
            &\phantom{=M} \mp \r_{i^{(1)}}\ldots \r_{i^{(n)}}\r_{n+1} \\
            &= {\color{mycolor1}\sum\limits_{i=1}^{n+1}\r_i}{\color{mycolor2} - \sum\limits_{\substack{i, i' = 1 \\ i< i'}}^{n+1} \r_i \r_{i'}} { \color{mycolor3} + \sum\limits_{\substack{ i,i',i''= 1 \\ i<i'< i''}}^{n+1} \r_i \r_{i'} \r_{i''}} - \ldots {\color{mycolor4} \pm \sum\limits_{\substack{i^{(1)},\ldots, i^{(n)}=1 \\ i^{(1)}<\ldots< i^{(n)}}}^{n+1} \r_{i^{(1)}} \ldots \r_{i^{(n)}}} \mp \r_{i^{(1)}}\ldots \r_{i^{(n)}}\r_{n+1}.
         \end{align*}
         This completes the proof.
    \end{description}
\end{proof}

\begin{lemma}\label{lemma:prod-sum}
Let $[\r_1, \ldots, \r_n] \in [0,1]^n$. Then
    \[1- \prod\limits_{i=1}^n (1-\r_i) = \sum\limits_{i=1}^{n}\r_i - \sum\limits_{\substack{i, i' = 1 \\ i< i'}}^n \r_i \r_{i'} + \ldots \pm \sum\limits_{\substack{i^{(1)},\ldots, i^{(n-1)}=1 \\ i^{(1)}<\ldots< i^{(n-1)}}}^n \r_{i^{(1)}} \ldots \r_{i^{(n-1)}} \mp \r_{i^{(1)}}\ldots \r_{i^{(n)}}. \]
\end{lemma}

\begin{proof}
    The proof proceeds via induction on n and uses the same observation as the previous proof. Again, we explicitly name the step $n=2$ as it uses a similar reasoning to the general case of $n\to n+1$ and thus makes the latter step easier to follow.
    \begin{description}
        \item[\underline{n=1:}] Clear.
        \item[\underline{n=2:}] \[ 1-\prod_{i=1}^2 (1-\r_i) = 1-(1-\r_1)(1-\r_2) = 1-(1-\r_1-\r_2+\r_1\r_2) = \r_1 + \r_2 - \r_1\r_2. \]
        \item[\underline{n$\to$ n+1:}] Let $1- \prod\limits_{i=1}^n (1-\r_i) = \sum\limits_{i=1}^{n}\r_i - \sum\limits_{\substack{i, i' = 1 \\ i< i'}}^n \r_i \r_{i'} + \ldots \pm \sum\limits_{\substack{i^{(1)},\ldots, i^{(n-1)}=1 \\ i^{(1)}<\ldots< i^{(n-1)}}}^n \r_{i^{(1)}} \ldots \r_{i^{(n-1)}} \mp \r_{i^{(1)}}\ldots \r_{i^{(n)}}$\\
        Which is equivalent to\\
        $\prod\limits_{i=1}^n (1-\r_i) = 1 -  \underbrace{\left(\sum\limits_{i=1}^{n}\r_i - \sum\limits_{\substack{i, i' = 1 \\ i< i'}}^n \r_i \r_{i'} + \ldots \pm \sum\limits_{\substack{i^{(1)},\ldots, i^{(n-1)}=1 \\ i^{(1)}<\ldots< i^{(n-1)}}}^n \r_{i^{(1)}} \ldots \r_{i^{(n-1)}} \mp \r_{i^{(1)}}\ldots \r_{i^{(n)}}\right)}_{=:y}$.
        With this we obtain
        \begin{align*}
            & 1- \prod\limits_{i=1}^{n+1} (1-\r_i) \\
            =& 1-\prod\limits_{i=1}^n (1-\r_i)\times(1-\r_{n+1}) = 1-[(1-y) (1-\r_{n+1})] \\
            =& 1-[1-y-\r_{n+1}+y\r_{n+1}] = y + \r_{n+1} - y\r_{n+1} \\
            =& {\color{mycolor1}\sum\limits_{i=1}^{n}\r_i}  {\color{mycolor2} - \sum\limits_{\substack{i, i' = 1 \\ i< i'}}^n \r_i \r_{i'}} { \color{mycolor3} + \sum\limits_{\substack{ i,i',i''= 1 \\ i<i'< i''}}^n \r_i \r_{i'} \r_{i''}} - \ldots {\color{mycolor4}\pm \r_{i^{(1)}}\ldots \r_{i^{(n)}}}  {\color{mycolor1} + \r_{n+1}} \\
            &{\color{mycolor2} - \r_{n+1}\sum\limits_{i=1}^{n}\r_i } {\color{mycolor3} + \r_{n+1}\sum\limits_{\substack{i, i' = 1 \\ i< i'}}^n \r_i \r_{i'}}  - \ldots {\color{mycolor4} \pm \r_{n+1}\sum\limits_{\substack{i^{(1)},\ldots, i^{(n-1)}=1 \\ i^{(1)}<\ldots< i^{(n-1)}}}^n \r_{i^{(1)}} \ldots \r_{i^{(n-1)}}} \mp \r_{i^{(1)}}\ldots \r_{i^{(n)}}\r_{n+1}\\
            =& {\color{mycolor1}\sum\limits_{i=1}^{n+1}\r_i}{\color{mycolor2} - \sum\limits_{\substack{i, i' = 1 \\ i< i'}}^{n+1} \r_i \r_{i'}} { \color{mycolor3} + \sum\limits_{\substack{ i,i',i''= 1 \\ i<i'< i''}}^{n+1} \r_i \r_{i'} \r_{i''}} - \ldots {\color{mycolor4} \pm \sum\limits_{\substack{i^{(1)},\ldots, i^{(n)}=1 \\ i^{(1)}<\ldots< i^{(n)}}}^{n+1} \r_{i^{(1)}} \ldots \r_{i^{(n)}}} \mp \r_{i^{(1)}}\ldots \r_{i^{(n)}}\r_{n+1}. 
        \end{align*}
        This completes the proof. \qed
    \end{description}
\end{proof}

\noindent The theorem now follows immediately from \Cref{lemma:lin-fnctn} and \Cref{lemma:prod-sum}.

\section{Axiom Compliance for Outcome Responsibility Functions}\label{app:outcome_proofs}

Here we present the proofs omitted in \Cref{sec:outcomeresp}.

\NURvsNRV*
\begin{proof}
To prove result (i), consider the coordination game depicted in \Cref{fig:coordination_rescue} and assume that the responsibility function $\R$ satisfies (KSym).
Towards a contradiction assume that $\R$ satisfies (FMC) and (NUR) in addition to (KSym). Again, by (FMC) we have $\R(w_2, j) = \R(w_3, j) = 1$ and by (KSym) it additionaly has to hold that $\R(w_1, j) = \R(w_4, j) = 1$. But this leaves agent $j$ without a strategy to avoid responsibility and thus violates (NUR).
Concerning (ii), consider the adapted decision tree in \Cref{fig:coordination_with_nature}. By symmetry of the tree, (NUR) and (KSym) agent $i$ gets assigned responsibility of 0 for all four outcome nodes. 
% By (NIRV) this yields some non-zero amount of responsibility for agent $i$ for both negative outcomes. 
This violates (NRV) as no agent or group of agents is assigned any responsibility.
\end{proof}

\renewcommand{\thefigure}{4}
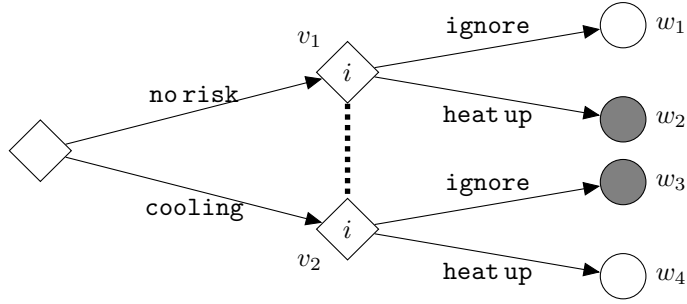
\begin{figure}
\centering
% \subfigure{
    \begin{tikzpicture} [-, scale=1, every node/.style={scale=1}, node distance=1cm, decision node/.style={diamond, draw, aspect=1, minimum height= 7mm},  prob node/.style={regular polygon,regular polygon sides=4, draw, minimum height= 1cm},  dummy/.style = {}, good outcome/.style={circle, draw, minimum width=6mm}, bad outcome/.style={circle, draw, fill=gray, minimum width=6mm }, >= triangle 45]
\node[decision node] (1) [] [] {\phantom{$i$}};

\node[dummy] (a) [right of = 1] [] {};
\node[dummy] (a1) [right of = a] [] {};
\node[dummy] (a2) [right of = a1] [] {};

\node[decision node] (2) [above right =  of a2] [label=above left: $v_1$] {$i$};
\node[decision node] (3) [below right = of a2] [label=below left: $v_2$] {$i$};

\node[dummy] (b) [right of = 2] [] {};
\node[dummy] (b1) [right of = b] [] {};
\node[dummy] (b2) [right of = b1] [] {};
\node[dummy] (b3) [right of = b2] [] {};

\node[good outcome] (good2) [above right = .4cm of b2] [label = right: $w_1$]{};
\node[bad outcome] (bad2) [below right = .4cm of b2] [label = right: $w_2$] {};

\node[dummy] (c) [right of = 3] [] {};
\node[dummy] (c1) [right of = c] [] {};
\node[dummy] (c2) [right of = c1] [] {};

\node[bad outcome] (bad3) [above right = .4cm of c2] [label = right: $w_3$] {};
\node[good outcome] (good3) [below right = .4cm of c2] [label = right: $w_4$] {};

\path[->] (1) edge node[above] {$\mathtt{no\,risk}$}  (2) (1) edge node[below] {$\mathtt{cooling}$} (3) (2) edge node[above] {$\mathtt{ignore}$}  (good2) (2) edge node[below] {$\mathtt{heat\,up}$} (bad2) (3) edge node[below] {$\mathtt{heat\,up}$} (good3) (3) edge node[above] {$\mathtt{ignore}$} (bad3);

\draw[line width=2pt, dotted=on] (2) -- (3); 

\end{tikzpicture}

\caption{\label{fig:coordination_with_nature} An adaption of the coordination game from \Cref{fig:coordination_rescue}. Agent $i$ can select whether to induce global warming, not knowing whether there is a threat of an impending ice age that this would alleviate or not, in which case the warming leads to a climate crisis.}
\end{figure}

\evalcombs*
\begin{proof}
Note that (NRV) is implied by (NIRV). The coordination game depicted in \Cref{fig:coordination_rescue} shows that $\Rneg$ does not satisfy (NRV) or (NIRV) since both agents get assigned responsibility 0 in all their decision nodes no matter what action they take. To prove (NIRV) for the other responsibility functions we single out the decision node $v \in V_d$ that is closest to some $w \in \epsilon$ while still achieving $B(v) \nsubseteq \epsilon$. Each of the remaining two responsibility functions assigns positive responsibility to the agent that corresponds to that decision node.
Regarding (NUR) it is easy to check for $\Rlike$ and $\Rneg$ that an agent always has an action available with zero responsibility. The exception is $\Rrisk$, which satisfies (FMC) and (KSym) and thus by \Cref{thm:NURvsNRV} can not satisfy (NUR).
Concerning (CC), it can be shown for all three responsibility functions that they do not assign responsibility along the path of the strategy that leads to $w \in V_o \setminus \epsilon$.

\end{proof}

\section{Computation of Outcome Responsibility Values}

Recall the example.

\renewcommand{\thefigure}{3}
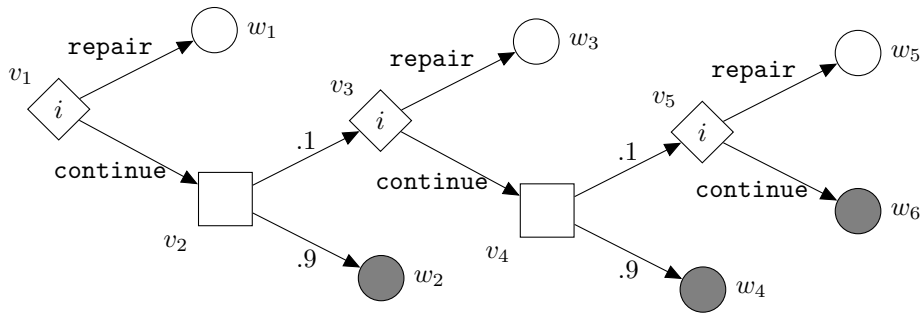
\begin{figure}
\centering
\begin{tikzpicture} [->, scale=1, every node/.style={scale=1}, node distance=1cm, decision node/.style={diamond, draw, aspect=1, minimum height= 7mm},  prob node/.style={regular polygon,regular polygon sides=4, draw, minimum height= 1cm},  dummy/.style = {}, good outcome/.style={circle, draw, minimum width=6mm}, bad outcome/.style={circle, draw, fill=gray, minimum width=6mm }, >= triangle 45]
\node[decision node] (v1) [] [label=above left: $v_1$] {$i$};

\node[dummy] (a) [right of = 1] [] {};
%\node[dummy] (a2) [right of = a] [] {};

\node[good outcome] (w1) [above right =  of a] [label=right: $w_1$] { };
\node[prob node] (p1) [below right = of a] [label=below left: $v_2$] { };

\node[dummy] (b) [right of = p1] [] {};
%\node[dummy] (b2) [right of = b] [] {};
%\node[dummy] (b3) [right of = b2] [] {};

\node[decision node] (v2) [above right =  of b] [label=above left: $v_3$]{$i$};
\node[bad outcome] (w2) [below right =  of b] [label=right: $w_2$] {};

\node[dummy] (c) [right of = v2] [] {};
%\node[dummy] (c2) [right of = c] [] {};

\node[good outcome] (w3) [above right = of c] [label=right: $w_3$] {};
\node[prob node] (p2) [below right = of c] [label=below left: $v_4$] {};

\node[dummy] (d) [right of = p2] [] {};
%\node[dummy] (c2) [right of = c] [] {};

\node[decision node] (v3) [above right = of d] [label=above left: $v_5$] {$i$};
\node[bad outcome] (w4) [below right = of d] [label=right: $w_4$] {};

\node[dummy] (e) [right of = v3] [] {};
%\node[dummy] (c2) [right of = c] [] {};

\node[good outcome] (w5) [above right = of e] [label=right: $w_5$] {};
\node[bad outcome] (w6) [below right = of e] [label=right: $w_6$] {};

\path (v1) edge node[above] {$\mathtt {repair \qquad}$} (w1) (v1) edge node[below] {$\mathtt{continue \, \qquad}$} (p1) (p1) edge node[above] {$.1$}  (v2) (p1) edge node[below] {$.9$} (w2) (v2) edge node[above] {$\mathtt {repair \qquad}$} (w3) (v2) edge node[below] {$\mathtt{continue \, \qquad}$} (p2) (p2) edge node[above] {$.1$} (v3) (p2) edge node[below] {$.9$} (w4) (v3) edge node[above] {$\mathtt {repair \qquad}$} (w5) (v3) edge node[below] {$\mathtt{continue \, \qquad}$} (w6);

\end{tikzpicture}
\caption{Scenario of a repeated action. The operator of a large machine can either ignore a warning of the machine and continue using it, knowing that this has a high chance of irreversible damage, or perform a repair. At some point in time the last chance to repair the machine is reached.}
\end{figure}

Let $\T$ as in the example and $G = \{i\}$. We first need to compute the values for $\rrisk$ and $\rneg$ for every action in each decision node.

\begin{align*} 
\rrisk  (\{i\}, i, v_1, \mathtt{repair}) &= \max\limits_{\zeta \in Z^\sim(\T, \{i\}, v_1)} \Delta \omega(v_1, \zeta, \mathtt{repair}) = \max\limits_{\zeta \in Z^\sim(\T, \{i\}, v_1)} [\omega (w_1, \zeta) - \omega(v_1, \zeta)]\\
&= \max\limits_{\zeta \in Z^\sim(\T, \{i\}, v_1)} [ \min_{\sigma\in\Sigma(\T, \{i\}, w_1)} \ell(\epsilon | w_1,\sigma,\zeta) - \min_{\sigma\in\Sigma(\T, \{i\}, v_1)} \ell(\epsilon | v_1,\sigma,\zeta)]\\
&= 0 - 0 = 0
\end{align*}

\begin{align*} 
\rrisk  (\{i\}, i, v_1, \mathtt{continue}) &= \max\limits_{\zeta \in Z^\sim(\T, \{i\}, v_1)} \Delta \omega(v_1, \zeta, \mathtt{continue}) = 0.9 - 0 = 0.9
\end{align*}

\begin{align*} 
\rrisk  (\{i\}, i, v_3, \mathtt{repair}) &= \max\limits_{\zeta \in Z^\sim(\T, \{i\}, v_3)} \Delta \omega(v_3, \zeta, \mathtt{repair}) = 0 - 0 = 0
\end{align*}

\begin{align*} 
\rrisk  (\{i\}, i, v_3, \mathtt{continue}) &= \max\limits_{\zeta \in Z^\sim(\T, \{i\}, v_3)} \Delta \omega(v_3, \zeta, \mathtt{continue}) = 0.9 - 0 = 0.9
\end{align*}

\begin{align*} 
\rrisk  (\{i\}, i, v_5, \mathtt{repair}) &= \max\limits_{\zeta \in Z^\sim(\T, \{i\}, v_5)} \Delta \omega(v_5, \zeta, \mathtt{repair}) = 0 - 0 = 0
\end{align*}

\begin{align*} 
\rrisk  (\{i\}, i, v_5, \mathtt{continue}) &= \max\limits_{\zeta \in Z^\sim(\T, \{i\}, v_5)} \Delta \omega(v_5, \zeta, \mathtt{continue}) = 1 - 0 = 1
\end{align*}

As the member contributions for selecting $\mathtt{repair}$ are always 0, i.e. $\underline\rho(G, i, v) = 0$, it holds that $\rrisk (G, i, v, \a) = \rneg (G, i, v, \a)$ for all $G\subseteq I$, $i\in G$, $v \in V_i$ and $\a \in A_v$ in this example.

Now we can use the member contribution values together with the aggregation function $\mathbf{mprod}$ to compute the outcome responsibility. 

\begin{align*}
    \mathbf{mprod} \circ \rrisk (\{i\}, i, w_1) &= \rrisk (\{i\}, i, v_1, \mathtt{repair}) = 0 = \mathbf{mprod} \circ \rneg (\{i\}, i, w_1)
\end{align*}

\begin{align*}
    \mathbf{mprod} \circ \rrisk (\{i\}, i, w_2) &= \rrisk (\{i\}, i, v_1, \mathtt{continue}) = 0.9 = \mathbf{mprod} \circ \rneg (\{i\}, i, w_2)
\end{align*}

\begin{align*}
    \mathbf{mprod} \circ \rrisk (\{i\}, i, w_3) &= 1- [ (1- \rrisk (\{i\}, i, v_1, \mathtt{continue})) (1- \rrisk (\{i\}, i, v_3, \mathtt{repair})] \\
    & = 1- [(1-0.9)(1-0)] = 1-[0.1] = 0.9 \\
    & =  \mathbf{mprod} \circ \rneg (\{i\}, i, w_3)
\end{align*}

\begin{align*}
    \mathbf{mprod} \circ \rrisk (\{i\}, i, w_4) &= 1- [ (1- \rrisk (\{i\}, i, v_1, \mathtt{continue})) (1- \rrisk (\{i\}, i, v_3, \mathtt{continue})] \\
    & = 1- [(1-0.9)(1-0.9)] = 1-[0.01] = 0.99 \\
    & =  \mathbf{mprod} \circ \rneg (\{i\}, i, w_4)
\end{align*}

\begin{align*}
    \mathbf{mprod} \circ \rrisk (\{i\}, i, w_5) &= 1- [ (1- \rrisk (\{i\}, i, v_1, \mathtt{continue})) \\
    & \phantom{= 1- [ (1-} \cdot (1- \rrisk (\{i\}, i, v_3, \mathtt{continue}) (1-\rrisk(\{i\}, i, v_5, \mathtt{repair}))] \\
    & = 1- [(1-0.9)(1-0.9)(1-0)] = 1-[0.01] = 0.99 \\
    & =  \mathbf{mprod} \circ \rneg (\{i\}, i, w_5)
\end{align*}

\begin{align*}
    \mathbf{mprod} \circ \rrisk (\{i\}, i, w_6) &= 1- [ (1- \rrisk (\{i\}, i, v_1, \mathtt{continue})) \\
    & \phantom{= 1- [ (1-} \cdot (1- \rrisk (\{i\}, i, v_3, \mathtt{continue}) (1-\rrisk(\{i\}, i, v_5, \mathtt{continue}))] \\
    & = 1- [(1-0.9)(1-0.9)(1-1)] = 1-[0] = 1 \\
    & =  \mathbf{mprod} \circ \rneg (\{i\}, i, w_6)
\end{align*}

\end{document}